\documentclass[reqno,nofootinbib,12pt]{revtex4}
\usepackage[centertags]{amsmath}
\usepackage{amsfonts}
\usepackage{amssymb}
\usepackage{amsthm}
\usepackage{newlfont}
\usepackage{stmaryrd}
\usepackage{mathrsfs}
\usepackage{mathtools}
\usepackage{euscript}
\usepackage{graphicx}
\usepackage{enumerate}
\usepackage{todonotes}
\usepackage{comment}

\usepackage{graphicx}
\usepackage{epsfig}

\usepackage[utf8]{inputenc}
\newtheorem{thm}{Theorem}[section]

\newtheorem{lemma}{Lemma}[section]
\newtheorem{prop}{Proposition}[section]
\newtheorem{Corollary}{Corollary}[section]
\newtheorem{definition}{Definition}[section]

\theoremstyle{remark}

  \let\de=\delta

\newcommand{\opunit}{\text{1}\kern-0.22em\text{l}}

\DeclareMathAlphabet{\mathpzc}{OT1}{pzc}{m}{it}

\newcommand{\id}{\textrm{d}}

\def\bea{\begin{eqnarray}}
\def\eea{\end{eqnarray}}
\def\ba{\begin{array}}
	\def\ea{\end{array}}

\usepackage{changes}
\begin{document}
	
    \title{Stabilization in the eye of a cyclone}
	\author{Thibaut Demaerel}
	\affiliation{Instituut voor Theoretische Fysica, KU Leuven, Belgium}
	\author{Christian Maes}
	\affiliation{Instituut voor Theoretische Fysica, KU Leuven, Belgium}
		\author{Karel Neto\u{c}n\'{y}}
	\affiliation{Institute of Physics AS CR, Czech republic}

\begin{abstract}
We consider the systematic force on a heavy probe induced by interaction with an overdamped diffusive medium where particles undergo a rotating force around a fixed center.  The stiffness matrix summarizes the stability of the probe around that center, where the induced force vanishes. We prove that the introduction of the rotational force in general enhances the stability of that point (and may turn it from unstable to stable!), starting at second-order in the nonequilibrium amplitude. When the driving is further enhanced the stabilization occurs for a wide range of rotation profiles and the induced stiffness converges to a universal expression proportional to the average mechanical stiffness. The model thus provides a rigorous example of stabilization of a fixed point due to contact with a nonequilibrium medium and beyond linear order around equilibrium.
\end{abstract}
\maketitle
%\tableofcontents
%\newpage

\section{Introduction}
It is of much recent interest to understand the stabilization of fixed points for the dynamics of slow particles in contact with a fast nonequilibrium medium. Applications vary from the induced force on probes in contact with biological tissue or filaments to the motion of dust in atmospheric dynamics. While the context of probes in contact with nonequilibrium media is clearly physically interesting there are very few mathematical treatments and even less rigorous results \cite{mal,yan,pre,kaf,tim,epl,prl, Sheshka,njp}.  
The present paper  presents a mathematical study about the stabilization of a probe (the slow particle) 
which interacts with (fast) medium particles that are subject to a vortex-shaped force-field.  The induced force on the probe is zero in the cyclone's eye explaining the title of the paper, i.e., in the center of rotation which is a symmetry point for the medium dynamics.
Simulations and heuristics have been given in the paper \cite{njp}; here we add new results combined with a mathematically rigorous treatment of previous questions.  In words, we are dealing with a two-dimensional overdamped medium of diffusing particles in short range interaction with a probe.  Each medium particle moves under a sufficiently confining self-potential and undergoes a rotational force.  We consider the quasi-static regime in which the mechanical force on the probe is averaged over the (mostly unknown) stationary distribution of the medium.  That mean force on the probe vanishes when the probe is placed at the center of rotation.  We define the stiffness matrix to discuss the linear stability of that fixed point.   It allows to discuss the stabilizing effect of the rotations, or how the effective spring constant between origin and probe changes as function of the rotation amplitude.\\

A more formal introduction to our results is to consider an overdamped dynamics of a medium particle with position $y(t) \in D(0,R)\subset \mathbb{R}^2$ on a disk of radius $R$ centered at the origin (possibly $R=+\infty$),
\begin{equation}\label{med}
 \dot{y}_t = -\nabla_y U(x,y_t) - \epsilon\, F(y_t) + \sqrt{2 \beta^{-1}}\; \xi_t
\end{equation}
where $(\xi_t)$ stands for standard white noise with prefactor $\beta^{-1}$ equal to the temperature, representing the thermal environment of the medium. The  $x\in \mathbb{R}^2$ is the position of a probe in interaction with the medium through a smooth potential $U(x,y)$.  The latter contains also the self-potential on the medium particles.  The nonequilibrium driving of strength $\epsilon$ is in the force $F(y) = \omega(|y|)\,(-y_2,\,y_1)=|y|\,\omega(|y|)\,\hat{\varphi} $ which is a divergence-less (possibly differential) rotation around the origin. The unit vector $\hat{\varphi} $ is in the azimuthal direction and the rotation profile $r \mapsto \omega(r)$ is assumed to be smooth in the radial coordinate $|y|=r$. For $\epsilon=0$ the dynamics is reversible at inverse temperature $\beta$ with equilibrium density $\rho_\text{eq}(y) = \exp[-\beta U(x,y)]/Z_x$ with respect to $\id y$.  In the next section we give mathematical details and sufficient assumptions for the well-behavedness of the Markovian dynamics corresponding to \eqref{med}. There will be a unique smooth stationary probability measure $\rho_{x,\epsilon}(y) \id y$ and we will discuss  the `smoothness' of the map $x \mapsto \rho_{x,\epsilon}(y)$, which is necessary  to define the stiffness matrix ${\cal K}$ from the mean force $x \mapsto {\mathscr F}(x)$ exerted on the probe when $x\rightarrow 0$,
\[
{\cal K}_{ij}(\epsilon) := -\frac{\partial \mathscr{F}_j}{\partial x^i}(0,\epsilon),\qquad
\mathscr{F}_i(x,\,\epsilon) := -\left\langle \frac{\partial U}{\partial x^i}(x,y)\right\rangle_{x,\,\epsilon},\quad i=1,2
\]
The expectation in $\langle\cdot\rangle_{x,\,\epsilon}$ is the average under $\rho_{x,\epsilon}(y)\,\id y$.  The force ${\mathscr F}(x)$ is zero for $x=0$ so that the stiffness estimates the stability of the probe at the origin.
The entry ${\cal K}_{ii}$ on the diagonal of the stiffness matrix can be called the spring constant $\kappa(\epsilon) = \kappa(-\epsilon)$ which starts as $\epsilon^2$ around equilibrium (where $\epsilon=0$).\\

 After the preliminary material of Section \ref{mod} we start in Section \ref{stifm} by proving a general expression for the stiffness matrix which distinguishes clearly the nonequilibrium effects.
 What we show is that
 \begin{equation}\label{difnon}
 {\cal K}_{ij}(\epsilon) =
 {\cal K}_{ij}(0) + \beta\epsilon\, \int_0^{+\infty}\id t\,\, \left\langle\frac{\partial U}{\partial x^i}(0,y_0) F^k(y_0)\frac{\partial }{\partial y^k}\frac{\partial U}{\partial x^j}(0,y_t)\right\rangle_0
 \end{equation}
 where the time-correlation is in the nonequilibrium medium dynamics with the probe at $x=0$ and started at time zero with probability density $\rho_0(y)\propto \exp\left(-\beta U(0,y)\right)$.\\
We also prove the universality of the stiffness in the limit $|\epsilon|\rightarrow \infty$  where the stiffness becomes independent of the imposed rotation profile, getting
  \begin{equation}\label{inflim}
{\cal K}_{ij}(\infty):=   \lim_{\epsilon\uparrow \infty} {\cal K}_{ij}(\epsilon) =\int \id y\,\rho_{0}(y)\; \frac{\partial^2 U}{ \partial x^i\partial x^j}(0,y)=\left\langle \frac{\partial^2 U}{ \partial x^i\partial x^j}(0)\right\rangle_{0}
   \end{equation}
 Clearly then the asymptotic spring constant depends only on the mechanical stiffness in that limit while the difference
   \begin{eqnarray*}
{\cal K}_{ij}(\infty) - {\cal K}_{ij}(0)  &=&   \beta \left\langle\frac{\partial U}{\partial x^i}(0),\,\frac{\partial U}{\partial x^j}(0)\right\rangle_0\\
&=& \beta \left\{\left\langle\frac{\partial U}{\partial x^i}(0)\,\frac{\partial U}{\partial x^j}(0)\right\rangle_0-\left\langle\frac{\partial U}{\partial x^i}(0)\right\rangle_0\left\langle\frac{\partial U}{\partial x^j}(0)\right\rangle_0\right\}\geq 0
   \end{eqnarray*}
  is non-negative and picks up the equilibrium fluctuation in terms of the force-force covariance.\\
    
In Section \ref{smd} we show that the stiffness matrix is a smooth function of the driving parameter $\epsilon$ so that we can carry out expansions around equilibrium, $\epsilon=0$.  That is used in Section \ref{closetoeq} to prove that for either unequivocally attractive or for repulsive probe-particle interaction $U$ and non-negative rotation profile $\omega$ one has a relative stabilization, i.e.
\begin{equation*}
\frac{\id^2{\cal K}_{ij}}{\id\epsilon^2}(\epsilon=0)\xi^i\xi^j >0,\qquad \xi \in \mathbb{R}^2
\end{equation*}
(Einstein summation is assumed throughout the paper. Indices are always lowered or raised with respect to the standard metric $\delta_{ij}$).\\
In Section \ref{monot} we show the monotonicity in $\epsilon$ of the symmetric part of ${\cal K}_{ij}(\epsilon)$ when the rotation is rigid ($\omega(r)=1$), even when an interaction between the medium particles is added which only depends on the relative distance between the particles. That means that the spring constant $\kappa(\epsilon)$ is non-decreasing in $\epsilon$.  We expect (but do not prove) that that scenario of monotone stabilization of the probe at the center as function of the rotational amplitude holds true under much weaker conditions as well\footnote{Counterexamples exist however, but are not discussed here.}.

\section{Model and mathematical setup}\label{mod}
The present section contains mathematical specifications concerning model \eqref{med}.  Some generalities from the theory of diffusion processes are repeated in the present context for self-consistency.  We emphasize mostly the construction of the pseudo-inverse of the backward generator in Section \ref{pse}.\\
We start with independent medium particles having positions $y$ on the open disk $M=D(0,R)\subset \mathbb{R}^2$ (possibly $R=+\infty$) and a probe is placed statically at a position $x \in M$.

\subsection{Forces}
We give here assumptions on the potential $U(x,y)= V_{pp}(|x-y|)+V_{ext}(|y|)$ and forces $F(y) = \omega(|y|)\,(-y_2,\,y_1)=|y|\,\omega(|y|)\,\hat{\varphi} $ appearing in \eqref{med}.\\
We assume $(x,y) \mapsto V_{pp}(|x-y|)$ and $y \mapsto V_{ext}(y)$ both are smooth maps. The same is then true for $(x,y) \mapsto U(x,y)$.
We write $V(y):=  U(x=0,\,y)$ and we suppose that  $\int_M\exp(-\beta V(y))\,\id y <  \infty$.  Moreover, the potential $V$ is bounded below and diverges at the boundary of the disk $M$, i.e., for every closed interval $[a,b]$,
\begin{equation} \label{pot top}
V^{-1}([a,b])\subset M \text{ is compact }
\end{equation}
We assume finally that the divergence at the boundary (when $R<\infty$) is sufficiently pronounced: there are $\delta,m,\gamma>0$ so that for all $y \in M$ with $|y|>R-\delta$, 
\begin{equation} \label{boundary condition}
(1- \frac{\gamma}{\beta})(\nabla V(y))^2 - \frac{1}{\beta}\nabla^2 V(y) \geq m
\end{equation}
That condition assures that there is sufficient compactness, so that the Poincar\'e inequality (see \eqref{Poincare inequality}) holds true.\\
Additionally, we assume
 \begin{eqnarray} 
&& \int_M \id y\,\exp(-\beta V(y)) (\nabla V(y))^2 < \infty \label{int 1}\\
&& \int_M \id y\,\exp(-\beta V) (\nabla_x U(x=0,y))^2 < \infty. \label{int 2}
\end{eqnarray}
We further assume there to be some $\epsilon>0$ s.t. $\forall x \in B(0,\epsilon)$
\begin{eqnarray}
&& k_1(x):=\sup_{z \in M:\,|z|<|x|}\left\|\nabla_y (U(z)-U(0))\right\|_{\infty}=\sup_{z \in M:\,|z|<|x|}\,\sup_{y \in M}\left|\nabla_x (U(z,y)-U(0,y))\right| < \infty \nonumber\\
&& k_2(x):=\sup_{z \in M:\,|z|<|x|}\left\|\left(-(\nabla_y V)\cdot \nabla_y+\beta^{-1}\nabla_y^2\right)(U(z)-U(0))\right\|_{\infty} <\infty \nonumber\\
&& \limsup_{x\to 0} \frac{k_1(x)}{|x|}, \,\limsup_{x\to 0} \frac{k_2(x)}{|x|}\,<\infty. \label{lip}
\end{eqnarray}
For the rotational driving we suppose that the force $F=(F_1,F_2)$ is bounded: there is a finite $C$ so that for all $y \in M$,
\begin{equation}\label{bounded driving}
|F(y)|=\sqrt{F_1^2(y)+F_2^2(y)}=|\omega(|y|)\,|y| < C.
\end{equation}
For e.g. rigid rotation, where $\omega(r)\equiv 1$, that implies we work on a  finite disk: $R<\infty$.\\
Note that by the foregoing constructions the driving field leaves the potential  invariant, i.e.,
\begin{equation} \label{invariance}
F_i(y) \frac{\partial V}{\partial y^i}(y)=F_i(y) \frac{\partial U}{\partial y^i}(0,y)=0
\end{equation}
Likewise,  for (weakly-)differentiable functions $u,v$  and for $\Omega:=-F^i \partial_i$,
\begin{equation}\label{omg} 
\int_M \id y\, \rho_0 F^i \partial_i \left(\frac{u^2}{2}\right) = -\int_M \id y\, \rho_0 \,u \,\Omega u = 0
\end{equation}

\subsection{Spaces} \label{spaces}
To  define the generator and the diffusive dynamics, we first set the function-analytic stage. Consider the weight function given by the probability density 
\[
\rho_0(y) := \frac 1{Z} \exp\left(-\beta V(y)\right), \quad Z =\int \id y\, \exp\left(-\beta V(y)\right) <\infty
\]
always with respect to the Lebesgue measure $\id y$ on $\mathbb{R}^2$. It is easy, by virtue of \eqref{invariance}, to show that $\rho_0$ solves the stationary Kolmogorov forward equation associated to \eqref{med} provided $x=0$ (but for $\epsilon$ arbitrary.)\\
We use the notation
\[
\langle u\rangle_0 := \frac 1{Z}\,\int \id y\, \exp\left\{-\beta V(y)\right\}\,u(y)\,= \int \id y\, \rho_0(y)\,u(y)
\]
for the expectation of a function $u: \mathbb{R}^2\rightarrow \mathbb{R}$.
We define $\mathscr{L}^2$ as the real Hilbert space of functions $u:M \to \mathbb{R}$ whose norm associated to the inner product
\begin{equation*}
 (u,v)_{\mathscr{L}^2} := \int \id y \,\rho_0 \,uv =: (u,v)
\end{equation*}
is finite. For example, \eqref{omg} says that 
\begin{equation} \label{antisymmetry}
(u,\Omega v)=-(\Omega u,v).
\end{equation}
%\\More notations to be used are
%\begin{equation*}
%\left(V_i,\,W_j\right):=\int \id y \,\rho_0 \, V_i W_j \qquad \left(V_i,\,W^i\right):=\int \id y \,\rho_0 \, V_i W^i
%\end{equation*}
%for vector fields $V$ and $W$. 
Next, we use the notation $1$ for the constant function $y\in M \mapsto 1$ and  $H := 1^{\perp} \subset \mathscr{L}^2$.
\\
Consider the bilinear form
\begin{equation*}
(u,v)_{H^0} := \int \id y \,\rho_0 \nabla u \cdot \nabla v.
\end{equation*}

Note that \eqref{pot top}, \eqref{boundary condition} imply that the smooth function $\eta:\mathbb{R} \to (1,+\infty):y\mapsto \exp(\gamma(V(y)-\min(V))$ has the properties
\begin{enumerate}
	\item $\lim_{|s| \to R}\eta(s)=+\infty$.
	\item There exists an $m>0$ and a compact subset $K \subset M$ so that: 
	\begin{equation} \label{spec}
	\left.\frac{1}{\eta}\left(\nabla V - \beta^{-1}\nabla^2\right)\eta\right|_{M\setminus K}\geq \gamma m
	\end{equation}
\end{enumerate}
It is shown in \cite{Bakry} that the existence of such a ``Lyapunov function" implies a so-called Poincar\'e inequality:
For all weakly differentiable functions $u:M \to \mathbb{R}$, we have the Poincar\'e inequality
\begin{equation} \label{Poincare inequality} 
c^2\left(\left\|u\right\|_{\mathscr{L}^2}^2 - (u,1)(1,u)\right)=c^2\left(\left\langle u^2\right\rangle_{0} - \left\langle u\right\rangle_{0}^2\right) \leq \left\langle (\nabla u)^2\right\rangle_{0} = \left\|u\right\|_{H^0}^2
\end{equation}
where the constant $c$ depends only on the temperature $\beta^{-1}$ and the potential $V$.\\
As a consequence,
\begin{equation*}
H^0 := \left\{u \subset \mathscr{L}^2 \left.\right| (u,u)_{H^0} < +\infty,\, (u,1)_{\mathscr{L}^2}=0\right\}
\end{equation*}
is a Hilbert space.

\subsection{Operators}\label{pse}
\begin{definition}
$\ell_{x,\epsilon}$ is the formal backward operator acting on twice-weakly differentiable functions $u:M\mapsto \mathbb{R}$ as
\[\ell_{x,\epsilon}u=-\left((\nabla_y U)(x,y)+\epsilon F\right) \cdot (\nabla u)(y)+\beta^{-1}(\nabla^2u)(y) \]
(Its co-domain being locally integrable functions).
\end{definition}
\begin{prop} For all $x$ in some open ${\cal W} \subset B(0,\epsilon)$
there exists a densely-defined unbounded operator $L_{x,\epsilon}:D \subset \mathscr{L}^2 \to \mathscr{L}^2$ with the following properties:
\begin{enumerate}
\item For $u\in H^0$ twice weakly-differentiable so that $\ell_{x,\epsilon}u \in \mathscr{L}^2$, one has
\begin{equation} \label{smooth}
L_{x,\epsilon}u=\ell_{x,\epsilon}u
\end{equation}
\item $L_{x,\epsilon}$ is the generator of a positive strongly-$\mathscr{L}^2$-continuous contraction semi-group $\left\{S(t):\mathscr{L}^2 \to \mathscr{L}^2\right\}_{t\geq 0}$ which is probability-conserving: $S(t)1 \equiv 1$.
\item For $x=0$ and all $u,v\in D$, abbreviating $L_{0,\epsilon}=:L_{\epsilon}$, one has
\begin{equation} \label{adj}
(u,L_{\epsilon}v)=(L_{-\epsilon}u,v)
\end{equation}
and
\begin{equation} \label{neg def}
(u,L_{\epsilon}u)=-\beta^{-1}(u,u)_{H^0}
\end{equation}
so that $L_{\epsilon}$ is negative--definite.
\item Further abbreviating $L_0=:L$, one has $L_{\epsilon}=L+\epsilon \Omega$ with $\Omega:H^0 \to \mathscr{L}^2: u \mapsto -F^i \partial_i u$ and for all $u,v \in H^0$,
\begin{equation*}
(u,Lv)=-\beta^{-1}(u,v)_{H^0}
\end{equation*}
\end{enumerate}
\end{prop}
Proofs of these statements can be found in the Appendix.

We now define the pseudo-inverse, as appearing in the following
\begin{prop}
There exists an operator $L_{\epsilon}^{-1}: \mathscr{L}^2 \to H^0 \subset \mathscr{L}^2$ with the following properties:
\begin{enumerate}
\item One has
\begin{equation*}
\forall u \in D:\, L_{\epsilon}^{-1}L_{\epsilon}u=P_Hu:=u-1(u,1), \qquad \forall u \in \mathscr{L}^2: L_{\epsilon}^{-1}u \in D \text{ and } L_{\epsilon}L_{\epsilon}^{-1}u=P_Hu
\end{equation*}
\item There are $k_3(x), k_4(x)>0$ so that $\forall x \in B(0,\epsilon)$
\begin{eqnarray} 
&&\sup_{u \in \mathscr{L}^2} \frac{\left\|L_{x,\epsilon}^{-1}u\right\|_{\mathscr{L}^2}}{\left\|u\right\|_{\mathscr{L}^2}}\leq k_3(x) \label{second bound} \\
&& \sup_{u \in \mathscr{L}^2}\frac{\left\|L_{x,\epsilon}^{-1}u\right\|_{H^0}}{\left\|u\right\|_{\mathscr{L}^2}}\leq k_4(x)\label{first bound}
\end{eqnarray}
\item For all $u,v \in \mathscr{L}^2$, 
\begin{equation}\label{pseudo adj}
(u,L_{\epsilon}^{-1}v)=(L_{-\epsilon}^{-1}u,v) \text{ and } (u,L_{\epsilon}^{-1}u)\leq 0
\end{equation}
\end{enumerate}
\end{prop}
Proofs of these statements care in the Appendix, subsection \ref{Kolm}, where the notation $R_{\epsilon} := L_{\epsilon}^{-1}$ is used.
%In the remainder of the paper, also the notation $(L+\epsilon \Omega)^{-1}$ is sometimes used as an alternative to $L_{\epsilon}^{-1}$.

\subsection{Stationary density: construction and uniqueness} \label{stationary density}
In this section, we construct the stationary density $\rho_{x,\epsilon}$ corresponding to probe position $x(\neq 0)$ and driving parameter $\epsilon$.

\begin{definition}
Let $x \in B(0,\epsilon)$. \\
Define the map (suppressing the $y$-dependence) $F_x:H^0\times \langle 1\rangle \to \mathscr{L}^2$
\begin{eqnarray*}
&&F_xu = -\exp(\beta V)\nabla\left[\exp(-\beta V)\cdot \nabla\left(U(x)- U(0)\right)u\right]\\
&&=\beta \ell_{0,0}(U(0)-U(x))u-\nabla\left( U(x)- U(0)\right)\cdot \nabla u \\
&&=\beta L(U(0)-U(x))u-\nabla\left( U(x)- U(0)\right)\cdot \nabla u
\end{eqnarray*}
\end{definition}
A few remarks:
\begin{enumerate}
\item Replacing $\ell_{0,0}$ by $L$, in the last step, is justified by \eqref{smooth} and assumption \eqref{lip}.
\item $F_x$ is well-defined with the stated (co-)domain since
\begin{eqnarray*}
&& \left\|F_x u\right\|_{\mathscr{L}^2} \leq \beta \left\|L(U(0)-U(x))\right\|_{\infty}\left\|u\right\|_{\mathscr{L}^2}+\left\|\nabla [ U(x)- U(0)]\right\|_{\infty}\left\|u\right\|_{H^0} \\
&&\leq \beta k_2(x)\left\|u\right\|_{\mathscr{L}^2}+k_1(x)\left\|u\right\|_{H^0} 
\end{eqnarray*}
If $u \in H^0$, the Poincar\'e inequality \eqref{Poincare inequality} implies that the RHS is smaller than $\left(\beta c^{-1}k_2(x)+k_1(x)\right)\left\|u\right\|_{H^0}=:k_F(x)\left\|u\right\|_{H^0}$ and $k_F(x)=O(x)$ by \eqref{lip}.
\end{enumerate}
\begin{prop} \label{sta}
In some open neighborhood of $x=0$, the function
\begin{equation} \label{stat dens}
\nu_{x,\epsilon}=\sum_{n=0}^{+\infty} (L_{-\epsilon}^{-1}F_x)^n 1
\end{equation}
is well-defined and $\nu_{x,\epsilon}-1\in \mathscr{L}^2\cap H^0$.
Moreover, $\nu_{x,\epsilon}$ is the unique element in $\mathscr{L}^2$ solving ($\forall u \in D$)
\begin{equation} \label{stat}
(\nu_{x,\epsilon},L_{x,\epsilon}u)=0.
\end{equation}
One also has
\begin{equation} \label{pert}
\nu_{x,\epsilon}=1+L_{-\epsilon}^{-1}F_x 1 + O(x^2)
\end{equation}
\end{prop}
Therefore $\rho_{x,\epsilon}:=\rho_0 \nu_{x,\epsilon}$ is the unique stationary density.\\
A proof of these statements can be found in the Appendix, subsection \eqref{proof}.
\section{The stiffness matrix}\label{stifm}
The statistical force $\mathscr{F}$ exerted on the probe when at position $x$ (in a neighborhood of the origin) and at driving $\epsilon$ is given by
\begin{equation*}
\mathscr{F}_i(x,\,\epsilon) := -\left\langle \frac{\partial U}{\partial x^i}(x)\right\rangle_{x,\,\epsilon}
\end{equation*}
That force is zero for $x=0$, due to the rotational symmetry of the system. To estimate the stability of the origin we introduce 
the stiffness matrix (at the origin), defined as
\begin{equation*}
{\cal K}_{ij}(\epsilon):=-\frac{\partial \mathscr{F}_j}{\partial x^i}(0,\epsilon)
\end{equation*}
It is not hard to show that rotational and reflection symmetry\footnote{Note that plane reflection switches the direction of the nonequilibrium force $\epsilon F \to -\epsilon F$. So formally, such a reflection switches the sign of the driving parameter $\epsilon$.} imply that the stiffness matrix takes the form
\begin{equation} \label{stm2}
{\cal K}(\epsilon)=\begin{pmatrix} 
m(\epsilon) & a(\epsilon) \\
-a(\epsilon) & m(\epsilon)
\end{pmatrix}
\end{equation}
where moreover $\forall \epsilon \in \mathbb{R}$
\begin{equation} \label{reflection law}
m(\epsilon)=m(-\epsilon), \qquad a(\epsilon)=-a(-\epsilon).
\end{equation}

\begin{thm}
The statistical force $x\mapsto \mathscr{F}(x,\epsilon)$ is differentiable at $x=0$ and the stiffness matrix equals
\begin{equation}
\label{stm}
{\cal K}_{ij}(\epsilon)=\left\langle \frac{\partial^2 U}{ \partial x^i\partial x^j}(0)\right\rangle_{0} -\beta \left(\frac{\partial U}{\partial x^i}(0),\,\frac{\partial U}{\partial x^j}(0)\right)+ \beta \left(\frac{\partial U}{\partial x^i}(0),\, \epsilon\Omega L_{\epsilon}^{-1}\frac{\partial U}{\partial x^j}(0)\right)
\end{equation}
or equivalently
\begin{equation} \label{stm3}
{\cal K}_{ij}(\epsilon)=\left\langle \frac{\partial^2 U}{ \partial x^i\partial x^j}(0)\right\rangle_{0}- \beta \left(\frac{\partial U}{\partial x^i}(0),\, L L_{\epsilon}^{-1}\frac{\partial U}{\partial x^j}(0)\right)
\end{equation}
\end{thm}

That expression distinguishes between the ``equilibrium'' and  the ``nonequilibrium'' contributions to the stiffness matrix as announced in \eqref{difnon}. The relation of the expression \eqref{stm} to the expression \eqref{inflim} (involving the time-integral) stems from the semigroup-identity $L_{\epsilon}^{-1}v=\int_0^{+\infty}S(t)v\id t$
which holds for all $v \in H$.
\begin{proof}
One has
\begin{eqnarray}
&& \left\langle \frac{\partial U}{\partial x^i}(x)\right\rangle_x -\left\langle \frac{\partial U}{\partial x^i}(0)\right\rangle_0 =\left\langle \left(\frac{\rho_{x,\epsilon}}{\rho_0}-1\right)\frac{\partial U}{\partial x^i}(x)\right\rangle_0 +\left\langle \frac{\partial U}{\partial x^i}(x)- \frac{\partial U}{\partial x^i}(0)\right\rangle_0\nonumber\\
&&=\left(\frac{\rho_{x,\epsilon}}{\rho_0}-1,\frac{\partial U}{\partial x^i}(x)\right)+\left(\frac{\partial U}{\partial x^i}(x)- \frac{\partial U}{\partial x^i}(0),1\right) \label{interm}
\end{eqnarray}
For the first term \eqref{interm}, \eqref{pert} and \eqref{lip} imply that
\begin{eqnarray}
&&\left(\frac{\rho_{x,-\epsilon}}{\rho_0}-1,\frac{\partial U}{\partial x^i}(x)\right)=\left(L_{-\epsilon}^{-1}F_x 1,\frac{\partial U}{\partial x^i}(0)\right)+O(x^2) \label{expl1}\\
&& = \beta\left(L_{-\epsilon}^{-1}L(U(0)-U(x)),\,\frac{\partial U}{\partial x^i}(0)\right)+O(x^2) \nonumber\\
&& = \beta\left(L(U(0)-U(x)),\,L_{\epsilon}^{-1}\frac{\partial U}{\partial x^i}(0)\right)+O(x^2) \label{expl3}\\
&& = \beta\left(U(0)-U(x),\,LL_{\epsilon}^{-1}\frac{\partial U}{\partial x^i}(0)\right)+O(x^2) \label{expl4}\\
&& = -\beta\left(\frac{\partial U}{\partial x^j}(0),\,LL_{\epsilon}^{-1}\frac{\partial U}{\partial x^i}(0)\right)x^j+O(x^2) \label{expl2}\\
&& = -\beta\left(\frac{\partial U}{\partial x^j}(0),\,\frac{\partial U}{\partial x^i}(0)\right)x^j+\beta\left(\frac{\partial U}{\partial x^j}(0),\,\epsilon\Omega L_{\epsilon}^{-1}\frac{\partial U}{\partial x^i}(0)\right)x^j+O(x^2) \nonumber
% \\&& = -\beta\left(\frac{\partial U}{\partial x^j}(0),\,\frac{\partial U}{\partial x^i}(0)\right)x^j-\beta\left(L_{-\epsilon}^{-1}\epsilon\Omega\frac{\partial U}{\partial x^j}(0),\, \frac{\partial U}{\partial x^i}(0)\right)x^j+O(x^2)\nonumber
\end{eqnarray}
where in \eqref{expl3} we use that, by the RHS of \eqref{int 1}, $\frac{\partial U}{\partial x}(x=0)$ is in $\mathscr{L}^2$ so that we can apply \eqref{pseudo adj}. In \eqref{expl4} we use that $L_{\epsilon}^{-1}$ maps to $D$ and \eqref{adj}. In \eqref{expl2} we use that by Lagrange's theorem
\[\left|U(0,y)-U(x,y)-\frac{\partial U}{\partial x^j}(0,y)x^j\right|=\left|\frac{\partial U}{\partial x^j}(x_y,y)-\frac{\partial U}{\partial x^j}(0,y)\right||x|\leq k_1(x_y)|x|\leq k_1(x)|x|\]
so that $\left\|U(0)-U(x)-\frac{\partial U}{\partial x^j}(0)\right\|_{\mathscr{L}^2}\leq \left\|U(0)-U(x)-\frac{\partial U}{\partial x^j}(0)\right\|_{\infty} =O(x^2)$.\\
Regarding the second term in \eqref{interm}, by \eqref{lip} we see that the map
\[\varphi:(0,1)\times M \to \mathbb{R}:(t,y)\mapsto \frac{1}{|x|t}\left(\frac{\partial U}{\partial x^i}(xt,y)- \frac{\partial U}{\partial x^i}(0,y)\right)\]
is bounded above by $q1$ for some $q>0$. Since this upper bound is integrable, the dominated convergence theorem implies that
\begin{equation*}
\lim_{t \to 0+}\frac{1}{t}\left\langle \frac{\partial U}{\partial x^i}(xt)-\frac{\partial U}{\partial x^i}(0)\right\rangle_{0} = \left\langle \frac{\partial^2 U}{ \partial x^k\partial x^i}(0)\right\rangle_{0}(x^k)+o(|x|).
\end{equation*}
Now one arrives at the desired expression by combining the two terms again.
\end{proof}

\section{Example: linear forces}\label{har}
As interludium  consider the example (introduced in \cite{njp}) where
\begin{enumerate}
\item The reservoir is $\mathbb{R}^2$, i.e., $R=\infty$.
\item The potentials are harmonic: $V_{ext}(y)=\frac{1}{2}k_1|y|^2$, $V_{pp}(y)=\frac{1}{2}k_2|y-x|^2$.
\item The driving is rigid: $F=(-y_2,y_1) \leftrightarrow \omega \equiv 1$.
\end{enumerate}
Note that while $1$ and $3$ are incompatible within the framework of this paper due to the constraint \eqref{bounded driving}, this example can be worked out in full details and it provides useful insight in what can be expected.\\
Under those circumstances, one can compute that the density
\begin{equation*}
\rho_{x,\epsilon}=\left(\frac{\beta (k_1+k_2)}{2\pi}\right) \exp\left(-\beta \frac{k_1+k_2}{2}(y-c_x)^2\right)
\end{equation*}
is stationary for
\[c_x = \frac{k_2}{(k_1+k_2)^2+\epsilon^2}\begin{bmatrix}
k_1+k_2& -\epsilon\\
\epsilon & k_1+k_2
\end{bmatrix}\begin{bmatrix}
x_1 \\
x_2
\end{bmatrix}.\]
Taking $x=(x_1,0)$, we can compute the components of the statistical force:
\begin{eqnarray*}
&& \mathscr{F}_1(x)=-\left\langle \frac{\partial V_{pp}}{\partial x_1}(x,y)\right\rangle_{x,\epsilon}= k_2\int_{\mathbb{R}^2}(y_1-x_1)\rho_{x,\epsilon}(y)\id y=k_2((c_x)_1-x_1)\int_{\mathbb{R}^2} \rho_{x,\epsilon}(y)\id y\\
&&=k_2((c_x)_1-x_1) = k_2\left[\frac{k_2(k_1+k_2)}{(k_1+k_2)^2 + \epsilon^2}-1\right]x_1 \\
&& \mathscr{F}_2(x)=-\left\langle \frac{\partial V_{pp}}{\partial x_2}(x,y)\right\rangle_{x,\epsilon}=k_2(c_x)_2=\left[\frac{k_2^2\epsilon}{(k_1+k_2)^2 + \epsilon^2}\right]x_1
\end{eqnarray*}
So we see that the radial force $f_1(x,\epsilon)$ becomes more center-directed as $|\epsilon|$ increases, saturating at a finite value $-k_2x_1$ as $|\epsilon| \to \infty$. A tangential component $f_2(x,\epsilon)$ is picked up due to the rotational force, yet this component vanishes again in the limit $\epsilon \to \infty$.\\
At the level of the stiffness matrix \eqref{stm2}, these calculations yield the following results for the on- and off-diagonal elements:
\[m(\epsilon)=-k_2\left[\frac{k_2(k_1+k_2)}{(k_1+k_2)^2 + \epsilon^2}-1\right], \qquad a(\epsilon)=\left[\frac{k_2^2\epsilon}{(k_1+k_2)^2 + \epsilon^2}\right].\]
We recognize the following qualitative properties:
\begin{enumerate}
\item $\epsilon \mapsto m(\epsilon)$ is a smooth increasing map and $\lim_{\epsilon \to \infty} m(\epsilon)=k_2$ exists.
\item $\epsilon \mapsto a(t)$ is smooth, zero at $\epsilon =0$ and zero as $\epsilon \to \infty$.
\item \eqref{inflim} is verified in this system (since $\frac{\partial^{2} V_{pp}}{\partial x^i \partial x^j}=k_2 \delta_{ij}$).
\end{enumerate}
The smoothness of these maps is a property which will be confirmed at the general level in Section \ref{smd}. The limiting behavior for large $|\epsilon|$ is again confirmed in appreciable generality as discussed in Section \ref{idl}. The positivity and monotonicity of $m$ does not hold in general. Positivity however holds for small $|\epsilon|$ when the probe exerts an unequivocally attractive or repulsive force on the medium and the profile $\omega \geq 0$, as proven in Section \ref{closetoeq}. Monotonicity is shown to hold for $\omega \equiv 1$, as shown in Section \ref{monot}.

\section{Smooth dependence on the driving strength}\label{smd}

\begin{lemma}
	$\forall u \in D$, $\forall m \in \mathbb{N}_0$ and $\forall \epsilon_0 \in \mathbb{R}$ we have
	\begin{equation} \label{smoothness 1}
	(L+\epsilon\Omega)^{-1}u =\sum_{k=0}^{m-1} (\epsilon_0-\epsilon)^k((L+\epsilon_0\Omega)^{-1}\Omega)^{k}(L+\epsilon_0\Omega)^{-1}u + \underbrace{(L-\epsilon\Omega)^{-1}((\epsilon_0-\epsilon) \Omega(L+\epsilon_0\Omega)^{-1})^{m} u}_{=:I_m(\epsilon)}
	\end{equation}
	and the $\mathscr{L}^2$-norm of $I_m(\epsilon)$ is $O(|\epsilon_0-\epsilon|^m)$.
\end{lemma}
\underline{Remark}: recall that $(L+\epsilon\Omega)^{-1}$ maps to $H^0$ which is in the domain of $\Omega$. Therefore the domain of the composition $\Omega(L+\epsilon\Omega)^{-1}$ is the entire $\mathscr{L}^2$-space and the expressions in the lemma above make sense without limitation.
\begin{proof}
	We prove \eqref{smoothness 1} by induction. The basis step is verified by
	\begin{eqnarray}
	&& \left\{(L+\epsilon\Omega)^{-1}-(L+\epsilon_0\Omega)^{-1}\right\} u = (L+\epsilon\Omega)^{-1}\left\{(L+\epsilon_0\Omega)-(L+\epsilon\Omega)\right\}(L+\epsilon_0\Omega)^{-1} u \nonumber\\
	&& =(\epsilon_0-\epsilon)(L+\epsilon\Omega)^{-1}\Omega(L+\epsilon_0\Omega)^{-1} u.
	\end{eqnarray}
	For the induction-step we have
	\begin{eqnarray*}
		&& I_m(\epsilon) -(\epsilon_0-\epsilon)^m(L+\epsilon_0\Omega)^{-1}(\Omega(L+\epsilon_0\Omega)^{-1})^{m+1}u \\
		&& = (\epsilon_0-\epsilon)^m\left\{(L+\epsilon\Omega)^{-1}-(L+\epsilon_0\Omega)^{-1}\right\}(\Omega(L+\epsilon_0\Omega)^{-1})^{m}u\nonumber \\
		&& = (\epsilon_0-\epsilon)^m(L+\epsilon\Omega)^{-1}\left\{(L+\epsilon_0\Omega)-(L+\epsilon\Omega)\right\}(L+\epsilon_0\Omega)^{-1}(\Omega(L+\epsilon_0\Omega)^{-1})^{m}u\\
		&& =(\epsilon_0-\epsilon)^{m+1}(L+\epsilon\Omega)^{-1}(\Omega(L+\epsilon_0\Omega)^{-1})^{m+1}u=I_{m+1}(\epsilon).
	\end{eqnarray*}
	The $\mathscr{L}^2$-norm of $I_{m}(\epsilon)$ can be bounded as
	\begin{equation*}
	\left\|I_{m}(\epsilon)\right\|\leq|\epsilon_0-\epsilon|^m \underbrace{\left\|(L+\epsilon\Omega)^{-1}\right\|}_{\leq k_3(0) \text{ by \eqref{second bound}}}\left\|\Omega((L+\epsilon_0\Omega)^{-1}\Omega)^mu\right\| \leq \text{ constant }|\epsilon_0-\epsilon|^m
	\end{equation*}
\end{proof}
Recall the expression \eqref{stm3} for the stiffness matrix.
writing $f_i:=\frac{\partial U}{\partial x^i}(0)$ and applying the expansion \eqref{smoothness 1}, we then get
\begin{prop} \label{smoothness}(smoothness of the stiffness matrix)
\newline
\begin{equation} \label{expres}
{\cal K}_{ij}(\epsilon)={\cal K}_{ij}(\epsilon_0)-\beta\sum_{k=1}^{m-1} (\epsilon_0-\epsilon)^k\left(f_i,\,L[L_{\epsilon_0}^{-1}\Omega]^{k}L_{\epsilon_0}^{-1}f_j\right) -\beta \left(f_i,\,LL_{\epsilon}^{-1}((\epsilon_0-\epsilon) \Omega L_{\epsilon_0}^{-1})^{m}f_j\right)
\end{equation}
implying that ${\cal K}_{ij}$ is a smooth function of $\epsilon$.
\end{prop}
In particular, when we set $\epsilon_0=0$. We get
\begin{equation} \label{special case}
{\cal K}_{ij}(\epsilon)={\cal K}_{ij}(0)-\beta\sum_{k=1}^{m-1} (-\epsilon)^k\left(f_i,\,[\Omega L^{-1}]^{k}f_j\right) -\beta \left(f_i,\,LL_{\epsilon}^{-1}((-\epsilon) \Omega L^{-1})^{m}f_j\right)
\end{equation}

\section{Vector fields and rotation invariance}
\begin{definition}
A radial function $\varphi:M \to \mathbb{R}$ is of the form $y \mapsto u(|y|)$ for some map $u:(0,+\infty) \to \mathbb{R}$. A vectorfield $V:M \to \mathbb{R}^2:y\mapsto(V_1(y),\,V_2(y))$ is called rotation invariant if a.e.
\begin{equation} \label{example}
\begin{cases}
& V_1(y)=\frac{1}{|y|}\left(V_r(y)y_1-V_{\varphi}(y)y_2\right)=V_r(y)\hat{y}_1-V_{\varphi}(y)\hat{y}_2\\
& V_2(y)=\frac{1}{|y|}\left(V_r(y)y_2+V_{\varphi}(y)y_1\right)=V_r(y)\hat{y}_2+V_{\varphi}(y)\hat{y}_1.
\end{cases}
\end{equation}
where $V_r$ and $V_{\varphi}$ are radial functions. When $V_{\varphi}$ vanishes in that expression, $V$ is a radial vectorfield.
\end{definition}
When we write e.g. $V \in \mathscr{L}^2$ or $V_i \in \mathscr{L}^2$ it is understood that all components of $V$ are in $\mathscr{L}^2$. The same agreement holds for $V \in H^0$, $V \in D$.\\
Let  $V,\,W$ be rotation invariant vector fields, then $\Omega V_{i}=\omega I_i^jV_j$, where
\begin{equation} \label{complex form}
I_i^j = \begin{pmatrix}
I_1^1 & I_1^2 \\
I_2^1 & I_2^2
\end{pmatrix} =
\begin{pmatrix}
0 & -1 \\
1 & 0
\end{pmatrix}
\end{equation}
Note that when $V \in D$, then $\Omega V\in D$.
When $V\in \mathscr{L}^2$ is rotation invariant and of the form \eqref{example}, Proposition \ref{equiv} (in the Appendix) implies that 
\begin{equation*}
\begin{cases}
& (L^{-1}V_1)(y)=(\tilde{R}V_r)(y)\hat{y}_1-(\tilde{R}V_{\varphi})(y)\hat{y}_2\\
& (L^{-1}V_2)(y)=(\tilde{R}V_r)(y)\hat{y}_2+(\tilde{R}V_{\varphi})(y)\hat{y}_1.
\end{cases}
\end{equation*}
where $\tilde{L}_{\epsilon}=L_{\epsilon}-\frac{1}{\beta r^2}$ maps radial functions to radial functions. So
$L_{\epsilon}^{-1}V:=(L_{\epsilon}^{-1}V_1,\,L_{\epsilon}^{-1}V_2)$ is rotation invariant as well. Likewise, when $V \in D$, $L_{\epsilon}V:=(L_{\epsilon}V_1,\,L_{\epsilon}V_2)$ is rotational invariant. 
\\
When $V \in \mathscr{L}^2$ and the rotation profile $r \mapsto \omega(r)\equiv \omega$ is constant, we have ($\forall \epsilon \in \mathbb{R}$)
\begin{equation*}
[\Omega,L_{\epsilon}^{-1}]V:=\Omega L_{\epsilon}^{-1}V - L_{\epsilon}^{-1}\Omega V=0
\end{equation*}
Moreover, when $V\in D$
\begin{equation*}
[\Omega ,L]V:=\Omega LV - L\Omega V=0. 
\end{equation*}

\section{Stabilization close to equilibrium} \label{closetoeq}
In this section, we prove a result detailing that when the probe-particle force is unequivocally attractive or repulsive and the rotation-profile $\omega \geq 0$, then close to equilibrium one has a relative stabilization of the stiffness matrix.
\begin{thm}
Suppose $y \mapsto f_i(y)y^i$ is either non-negative or non-positive. Suppose in addition that $y \mapsto \omega(|y|)$ is non-negative and $\omega f \neq 0$ in the sense of $\mathscr{L}^2$, then
\begin{equation} \label{close to eq}
\frac{\id^2{\cal K}_{ij}}{\id\epsilon^2}(\epsilon=0)\xi^i\xi^j >0
\end{equation}
for all $\xi\in \mathbb{R}^2$.
\end{thm}
\begin{proof}
Due to the rotational symmetry \eqref{reflection law}, it suffices to check the positivity of the trace
\begin{equation*}
\frac{\id^2{\cal K}_{i}^i}{\id\epsilon^2}(0).
\end{equation*}
Using expression \eqref{special case}, we see that
\begin{equation}
{\cal K}_{i}^i(\epsilon)={\cal K}_{i}^i(0)+\epsilon \underbrace{(f_i,\,\Omega L^{-1}f^i)}_{=0}-\epsilon^2 (f_i,\,(\Omega L^{-1})^2f^i)+o(\epsilon^2).
\end{equation}
Now
\begin{eqnarray*}&& (f_i,\,(\Omega L^{-1})^2 f^i)=\delta^{ij}(f_i,\,(\Omega L^{-1})^2f_j)=\delta^{ij}(f_i,\, \omega I_j^k L^{-1}\omega I_k^l L^{-1}f_l)\\
&&=-( f_i,\,\omega L^{-1} \omega L^{-1} f^i)=-( f,\,\omega \hat{y}_iL^{-1} \omega L^{-1} f^i)=-( f,\,\omega \tilde{R} (\omega \hat{y}_iL^{-1} f^i))\\
&&=-(f,\,\omega \tilde{R}\omega\tilde{R}(f^i \hat{y}_i))= -(f,\,\omega \tilde{R}\omega\tilde{R}f)
\end{eqnarray*}
where we write $f:=f^i\hat{y}_i$ and we use Proposition \ref{equiv} in the Appendix (recall that $L^{-1}=R_0$). But by item 4 of Proposition \ref{resolvent properties}, $-\tilde{R}$ maps non-zero, non-negative functions to positive functions. So $\tilde{R} f>0$. But then $\omega \tilde{R} f\geq0$ and non-zero. But then $\tilde{R}\omega \tilde{R} f>0$. Finally then, $\omega\tilde{R}\omega \tilde{R} f\geq0$ is positive in the support of $\omega$. Therefore, since $\omega f \neq 0$ in the sense of $\mathscr{L}^2$, $y \mapsto \omega(y)f(y)[\tilde{R}\omega \tilde{R} f](y)$ is a non-negative measurable function whose support has non-zero $\mathscr{L}^2$-measure. Therefore
\[\frac{\id^2}{\id\epsilon^2}{\cal K}_{i}^i(\epsilon)= -2( f_i,\,(\Omega L^{-1})^2f^i)=2( f,\,\omega\tilde{R}\omega \tilde{R}f)>0\]
\end{proof}

\section{Infinite driving limit}\label{idl}
Throughout this section, the bound \eqref{second bound} plays an important role.
\begin{lemma}
For every radial vectorfield $V$ for which both $V,\,\frac{V}{\omega} \in \mathscr{L}^2$, one has
\[ \|L_{\epsilon}^{-1}V\|_{\mathscr{L}^2} \leq \frac{1}{\epsilon}\left\|\frac{1}{\omega}V\right\|_{\mathscr{L}^2}.\]
\end{lemma}
\begin{proof}
Notice that for every radial vectorfield $W\in D$, one has
\begin{eqnarray*}
&& \left\|\frac{1}{\omega}L_{\epsilon}W\right\|_{\mathscr{L}^2}\geq \left(I_i^j \hat{y}_j ,\, \frac{1}{\omega}L_{\epsilon}W^i\right)=\left(I_i^j \hat{y}_j ,\, \frac{1}{\omega}L(w \hat{y}^i)\right)+\epsilon \delta^{il}\left(I_i^j \hat{y}_j ,\, I_l^kw \hat{y}_k\right) \\
&& = ( \underbrace{I_i^j\hat{y}_j\hat{y}^i}_{=0},\frac{1}{\omega}\tilde{L}(w)) +\epsilon \left( \hat{y}_i ,\, w \hat{y}^i\right)=\epsilon \left\|W\right\|_{\mathscr{L}^2}
\end{eqnarray*}
Plugging in $W=L_{\epsilon}^{-1}V$, which is indeed in $D$, we get $\left\|\frac{1}{\omega}V\right\|_{\mathscr{L}^2} \geq \epsilon\left\|L_{\epsilon}^{-1}V\right\|_{\mathscr{L}^2}$.
\end{proof}
Recall that the force $f(y):=-(\nabla_x U)(x=0,\,y)$ is a smooth and compactly supported vector-field. Therefore $f \in D$. So in the expression \eqref{stm3} for the stiffness-matrix, \eqref{adj} implies that
\[\left( f_i, \, LL_{\epsilon}^{-1}f_j\right)=\left( Lf_i, \, L_{\epsilon}^{-1}f_j\right).\]
We can thus establish \eqref{inflim} by virtue of the following corollary
\begin{Corollary}
Suppose the force on the probe $f=-(\nabla_x U)(x=0)$ and the rotation profile $\omega$ share the property that $\frac{f}{\omega} \in \mathscr{L}^2$, then
\begin{eqnarray}
\lim_{\epsilon \to \infty}\left( Lf_i, \, L_{\epsilon}^{-1}f_j\right)=0.
\end{eqnarray}
\end{Corollary}

\section{Monotonicity of the stabilization for rigid driving}\label{monot}
We say that the driving is rigid when $\omega$ is constant, say $\omega(r)\equiv 1$. The rotational driving is now $\Omega=\frac{\partial}{\partial \varphi}$.
The arguments that follow are based on the property that $[\Omega,L]V_i:=\Omega LV_i - L\Omega V_i=0$ for a radial vectorfield $V_i \in D$ (for which $\Omega V_i$ is still in $D$, so that the second term of the commutator makes sense).  That remains so when we add interaction between the medium particles which is rotation invariant, i.e. a two-body potential function of the distance between the medium particles.
\begin{lemma} \label{comm}
On all relevant operator domains, $[\Omega,L]=0$ implies that for all $\epsilon_1,\epsilon_2\in \mathbb{R}$
\begin{enumerate}
\item $L_{\epsilon_1}L_{\epsilon_2}=L_{\epsilon_2}L_{\epsilon_1}$
\item $L_{\epsilon_1}L_{\epsilon_2}^{-1}=L_{\epsilon_2}^{-1}L_{\epsilon_1}$
\item $L_{\epsilon_1}^{-1}L_{\epsilon_2}^{-1}=L_{\epsilon_2}^{-1}L_{\epsilon_1}^{-1}$
\end{enumerate}
\end{lemma}
\begin{proof}
The first identity is trivial. For the second identity, consider that for all $V \in D$,
\[L_{\epsilon_2}^{-1}L_{\epsilon_1}V = L_{\epsilon_2}^{-1}L_{\epsilon_1}L_{\epsilon_2}L_{\epsilon_2}^{-1}V=L_{\epsilon_2}^{-1}L_{\epsilon_2}L_{\epsilon_1}L_{\epsilon_2}^{-1}V=L_{\epsilon_1}L_{\epsilon_2}^{-1}V\]
where we used the first identity in the middle.
For the third identity, note that for all $V \in H$,
\begin{eqnarray*}
&& L_{\epsilon_1}^{-1}L_{\epsilon_2}^{-1}V= L_{\epsilon_1}^{-1}L_{\epsilon_2}^{-1}L_{\epsilon_1}L_{\epsilon_1}^{-1}V \\
&&=L_{\epsilon_1}^{-1}L_{\epsilon_1}L_{\epsilon_2}^{-1}L_{\epsilon_1}^{-1}V =L_{\epsilon_2}^{-1}L_{\epsilon_1}^{-1}V
\end{eqnarray*}
where we used the second identity in the middle.
\end{proof}
\begin{thm}
When $\Omega$ corresponds to rigid rotation, then the nonequilibrium contribution to the stiffness matrix is negative semi-definite and it is monotone in $|\epsilon|$.
\end{thm}
\begin{proof}
$f_i := \frac{\partial U}{\partial x^i}(x=0)$ is a rotation--invariant vector field, so all the identities of Lemma \eqref{comm} apply.\\
When $\epsilon=0$, the nonequilibrium contribution is zero, so it suffices to check the monotonicity in $|\epsilon|$.

According to the expansion \eqref{expres}, that monotonicity amounts to verifying that for all $\xi \in \mathbb{R}^d$ and $\epsilon>0$ we have
\begin{eqnarray*}
&& 0\geq\left(L_{-\epsilon}^{-1}LL_{-\epsilon}^{-1}\Omega f_i, \,f_j\right)\xi^i \xi^j = \left(L_{-\epsilon}^{-1}LL_{-\epsilon}^{-1}\Omega f_i\xi^i, \,f_j\xi^j\right) \\
&& = \frac{1}{2}\left(\left(L_{-\epsilon}^{-1}LL_{-\epsilon}^{-1}\Omega v, \,v\right)+\left(v,\,L_{-\epsilon}^{-1}LL_{-\epsilon}^{-1}\Omega v\right)\right) = I(\epsilon)
\end{eqnarray*}
where $v=f^j\xi_j$ has the same domain and symmetry properties as $f^j$.
\newline
Now,
\begin{eqnarray*}
&& 2I(\epsilon)=\left(L\Omega L_{-\epsilon}^{-1}L_{-\epsilon}^{-1} v, \,v\right)+\left(v,\,L\Omega L_{-\epsilon}^{-1}L_{-\epsilon}^{-1} v\right) \\
&& =\left( L_{-\epsilon}^{-1}L_{-\epsilon}^{-1} v, \,\left\{-\Omega L L_{-\epsilon}L_{-\epsilon}+ L_{\epsilon}L_{\epsilon} L\Omega\right\}L_{-\epsilon}^{-1}L_{-\epsilon}^{-1}v\right)
\end{eqnarray*}
and it suffices to check that the operator $-\Omega L L_{-\epsilon}L_{-\epsilon}+ L_{\epsilon}L_{\epsilon} L\Omega$ is negative--definite. After a calculation, one gets $4\Omega LL\Omega = -(2L\Omega)^{\dagger}(2L\Omega)$.
\end{proof}

\section{System with harmonic potentials}
We continue with the system with harmonic potentials as considered in Section \ref{har}:
\begin{enumerate}
\item $M=\mathbb{R}^2$,
\item $V_{pp}(|x-y|)=k_1|x-y|^2$ and $V_{ext}(|y|)=k_2 |y|^2$,\footnote{One can easily verify that these choices conform to all requirements on the potentials stated at the beginning of the paper.}
\item The rotation profile $r \mapsto \omega(r)$ is arbitrary (but conforms to \eqref{bounded driving} as always).
\end{enumerate}
What is intriguing about this system, is the fact that (the Cartesian components of) the probe-particle force at $x=0$
\begin{equation*}
f_j(y)=-\left.\left(\frac{\partial V_{pp}}{\partial y^j}(|x-y|)\right)\right|_{x=0}=-k_1y_j
\end{equation*}
is an eigenfunction of $L$. Indeed
\begin{equation*}
\left(-\frac{\partial V}{\partial y_i}(y)\frac{\partial}{\partial y^i}+\beta^{-1}\frac{\partial}{\partial y_i}\frac{\partial}{\partial y^i}\right)(f_j(y)) = -(k_1+k_2)f_j,
\end{equation*}
so that the corresponding eigenvalue $-(k_1+k_2)$ is found to be negative, if $\rho_0 \propto \exp\left(-\beta(k_1+k_2)|y|^2\right)$ is to correspond to a finite measure.
\\
To prove that ${\cal K}_i^i(\epsilon)-{\cal K}_i^i(0)\geq 0$, we invoke the following general \begin{prop}
${\cal K}_i^i(\epsilon)-{\cal K}_i^i(0)> 0$ whenever (every component of) the force $f_j$ is an eigenfunction of $L$. In that case, the precise form of the rotation profile $r \mapsto \omega(r)$ is immaterial.
\end{prop}
\begin{proof}
Recall from \eqref{neg def} that $L_{\epsilon}$ is negative-definite. Therefore, all nonzero eigenvalues $\lambda \in \mathbb{C}$ of $L_{\epsilon}$ must be real and negative.
Taking the trace of the stiffness matrix, we then get (for its nonequilibrium correction)
\begin{eqnarray*}
&& \left( (L-\epsilon\Omega)^{-1}\epsilon\Omega f_i,\, f^i\right)=\left( I_D(L-\epsilon\Omega)^{-1}\epsilon\Omega f_i,\, f^i\right)=\left( L^{-1}L(L-\epsilon\Omega)^{-1}\epsilon\Omega f_i,\, f^i\right) \\
&& =\left( L(L-\epsilon\Omega)^{-1}\epsilon\Omega f_i,\, L^{-1}f^i\right) \\
&& = \lambda^{-1}\left( L(L-\epsilon\Omega)^{-1}\epsilon\Omega f_i,\, f^i\right) \\
&& = \lambda^{-1}\left( (L-\epsilon\Omega+\epsilon\Omega)(L-\epsilon\Omega)^{-1}\epsilon\Omega f_i,\, f^i\right) \\
&& =\lambda^{-1}\epsilon\left(\underbrace{\left( \Omega f_i,\, f^i\right)}_{=0}+\epsilon \left( \Omega (L-\epsilon\Omega)^{-1}\Omega f_i,\, f^i\right)\right) \\
&& = -\lambda^{-1}\epsilon^2 \left( (L-\epsilon\Omega)^{-1}(\Omega f_i),\, (\Omega f^i)\right) \\
&& < 0.
\end{eqnarray*}
\end{proof}

\section{Conclusion}
The occurrence of stable patterns and fixed points for probes in contact with a driven medium is a long-standing challenge within nonequilibrium statistical mechanics.  In the present paper rigorous results have been given about the stabilization of a probe approaching the rotation center of an overdamped medium.  We have shown that the stiffness increases for small nonequilibrium driving and is positive for large rotational driving.  We expect but cannot prove that the same phenomenon of stabilization also occurs for a probe in contact with particles that flip the direction of their rotation at random times.\\

\noindent{\bf Acknowledgements}: 
TD has been partially supported by Belspo-IUAP ``Dynamics, Geometry and Statistical Physics.'' KN acknowledges the support from the Grant Agency of the Czech Republic, grant no. 17-06716S.

\section{Appendix}
\subsection{Construction of the Kolmogorov operators and the diffusion semi-groups} \label{Kolm}
We set out to construct the just-mentioned operators and semigroups by constructing the relevant resolvents first.

\begin{definition}
$H^{\mu}$ ($\mu>0$) is the Hilbert space resulting from the closure of the vector-space of radial weakly differentiable functions $u:M \mapsto \mathbb{R}$ whose norm associated to the inner product
\[(u,v)_{H^{\mu}}=\int_M \rho_0 \left(\nabla u \cdot \nabla v+\mu\frac{uv}{|y|^2}\right)\id y\]
is finite and $\limsup_{y \to 0}\frac{|u(y)|}{|y|}=0$. The closure being with respect to $(.,.)_{H^{\mu}}$.\\
Likewise $\hat{H}^{\mu}$ is the Hilbert space of weakly differentiable functions whose norm associated to the inner product
\[(u,v)_{\hat{H}^{\mu}}=(u,v)_{H^{\mu}}+(u,v)_{\mathscr{L}^2}\]
is finite.
\end{definition}
\begin{lemma} \label{isom lemma}
Let $u$, $v$ be radial $H^1$-functions, then $f^i,\,g^i:M\mapsto \mathbb{R}:y \mapsto u(y)\hat{y}^i$ resp. $v(y)\hat{y}^i$ are in $H^0$. Conversely, if $f$, $g$ are radial vectorfields with $H^0$-components, then $u,v:M \mapsto \mathbb{R}:y \mapsto f^i(y)\hat{y}_i,\,g^i(y)\hat{y}_i$ are in $H^1$. More precisely, one has the isometry
\begin{equation}\label{isometry}(u,v)_{H^1}=(u\hat{y}^i,\,v\hat{y}_i)_{H^0},\qquad (f^j\hat{y}_j,\,g^k\hat{y}_k)_{H^1}=(f^i,\,g_i)_{H^0}\end{equation}
For radial $H^0$-vectorfields $f$, $g$ and radial $H^1$-functions $u$, one also has
\begin{equation} \label{orth}
F^i\partial_i u = 0=g_jF^i\partial_i f^j.
\end{equation}
\end{lemma}
\begin{proof}
Notice that
\begin{equation}\label{form}\partial_i \hat{y}^j = \partial_i \left(\frac{y^j}{|y|}\right)= \frac{\delta_i^j}{|y|}-\frac{y_iy^j}{|y|^3}.\end{equation}
Then $\hat{y}^i(\partial_i \hat{y}^j)=0$ and $(\partial_i \hat{y}^j)(\partial^i \hat{y}_j)=\frac{1}{|y|^2}$ and it follows that
\begin{eqnarray*}
&& (f^j,\,g_j)_{H^0}=\int_M \partial_i (u(y)\hat{y}^j)\partial^i (v(y)\hat{y}_j) \id y \\
&& = \int_M \left\{ (\partial_i u)(y) (\partial^i v)(y) +\left[v(y)(\partial_iu)(y)+u(y)(\partial_iv)(y)\right]\underbrace{\hat{y}^j\left(\frac{\delta_j^i}{|y|}-\frac{y_jy^i}{|y|^3}\right)}_{=0}\varphi(y)+\frac{u(y)v(y)}{|y|^2}\right\}\id y \\
&& = \int_M \left\{ (\partial_i u)(y)(\partial^i v)(y) +\frac{u(y)v(y)}{|y|^2}\right\}\id y \\
&& = (u,\,v)_{H^1}
\end{eqnarray*}
which verifies the isometry \eqref{isometry}. In \eqref{orth}, the first equality is trivial. For the second, one has (writing $f^i=u \hat{y}^i$ with $u$ radial)
\[g_jF^i\partial_if^j = g_jF^i\partial_i(u\hat{y}^j)=g_juF^i \partial_i\hat{y}^j=0\]
by \eqref{form} and since $g_iF^i=0=y_iF^i$.
\end{proof}

\begin{definition}
	Let $\lambda \geq 0$ and define $B_{x,\epsilon,\lambda,\mu}:\hat{H}^{\mu} \times \hat{H}^{\mu} \to \mathbb{R}$ as the bilinear map which acts as
	\begin{equation*}
	(u,v) \mapsto \int \rho_0\left\{\left(\beta^{-1}\nabla u+\left(\nabla_y [U(x,y)- U(0,y)]+\epsilon F\right)u\right)\cdot\nabla v+\lambda \,uv+\beta^{-1}\mu\frac{uv}{|y|^2}\right\}\,\id y
	\end{equation*}
	We also use the following auxiliary notations:
\begin{itemize}
\item $B_{x,\epsilon,\lambda}:=B_{x,\epsilon,\lambda,\mu=0}$
\item $B_{\epsilon}=B_{x=0,\epsilon,\lambda=0}:=B_{x=0,\epsilon,\lambda=0,\mu=0}$
\item $\tilde{B}=B_{x=0,\epsilon=0,\lambda=0,\mu=1}$
\end{itemize}
.
\end{definition}
\begin{lemma} \label{bilinear operator} (properties of $B_{x,\epsilon,\lambda,\mu}$)
\newline
\begin{itemize}
\item For all $ u,v \in H^0$:
\begin{equation} \label{symmetry}
B_{0,\epsilon,\lambda,\mu}(u,v)=B_{0,-\epsilon,\lambda,\mu}(v,u).
\end{equation}
\item Let $u$ and $v$ be a pair of functions where one of the pair is smooth and compactly supported and the other is twice weakly differentiable, then
\begin{equation} \label{partial int}
(u,\,(\lambda +\frac{\mu}{\beta|y|^2}-l_{x,\epsilon})v)=B_{x,\epsilon,\lambda,\mu}(u,v).
\end{equation}
\end{itemize}
\end{lemma}
\begin{proof}
The first result is straightforward and it requires the global antisymmetry of $F$ as expressed in \eqref{antisymmetry}.
\eqref{partial int} follows from a single partial integration (where the definition of weak differentiation is crucial, as is the compact support of one of the two functions since this ensures that there is no boundary term in the partial integration.)
\end{proof}
\begin{lemma} \label{comp supp}
There exists an increasing sequence of smooth, non-negative, compactly supported functions $(\chi_n)_n$ converging to $1$ and for which $(\chi_n,\chi_n)_{H^0}\to 0$.
\end{lemma}
\begin{proof}
Let $\eta:\mathbb{R} \to [0,1]$ be some non-negative smooth function whose support is $(-1,1)$, which is increasing in $(-1,0)$ and decreasing in $(0,1)$ and which has $\eta(0)=1$.
Now, define the auxiliary sequence
\begin{equation*}
\eta_n(x)=\eta(x/n)
\end{equation*}
which is indeed increasing towards the constant function $x \mapsto 1$.
Finally, define $\chi_n = \eta_n \circ U(x=0)$. I.e. 
\begin{equation*}
\chi_n(y)=\eta_n(V(y))=\eta(V(y)/n).
\end{equation*}
The property that $\chi_n$ has compact support follows from \eqref{pot top}. Moreover
\begin{eqnarray}
&& (\chi_n,\chi_n)_{H^0}=Z^{-1}\int_{M} \exp(-\beta V(y)) (\partial_i \chi_n)(y)(\partial^i \chi_n)(y) \id y\nonumber\\
&& = \int_{M} \exp(-\beta V(y)) \left(\eta'(V(y)/n)\right)^2\frac{1}{n^2}(\partial_i V)(y)(\partial^i V)(y) \de y \nonumber\\
&& \leq \left(\frac{\|\eta'\|_{\infty}}{n}\right)^2 \int_{M} \exp(-\beta V(y)) (\partial_i V)(y)(\partial^i V)(y) \id y \nonumber\\
&& \leq C / n^2
\end{eqnarray}
where the constant in the end is supplied by \eqref{int 1}.
\end{proof}
The following calculation will turn out to be useful later:
Let $\chi\geq 0$ be a function with $F^i\partial_i\chi = 0$. Then \eqref{antisymmetry} implies
\begin{eqnarray}
&& B_{x,\epsilon,\lambda}(\chi u,u) \nonumber\\
&&= \int_M \rho_0\left(\beta^{-1}\nabla(\chi u)\cdot\nabla u+\chi u\nabla\left(U(x)-U(0)\right)\cdot\nabla u+\left[\lambda+\frac{\mu}{\beta|y|^2}\right] \chi u^2\right)\id y \label{inter2}\\
&& \geq \int_M \rho_0\left(\beta^{-1}[\chi(\nabla u)^2+u(\nabla \chi \cdot\nabla u)]-\chi k_1(x)|u\nabla u|+\left[\lambda+\frac{\mu}{\beta|y|^2}\right] \chi u^2\right)\id y \nonumber\\
&& \geq \int_M \rho_0\left(\left[\beta^{-1}(\chi+\frac{\theta}{2})+\frac{k_1(x)\chi}{2}\right](\nabla u)^2+\left[\lambda \chi+\frac{(\nabla \chi)^2}{2\theta}+\frac{\mu}{\beta|y|^2}\chi -\frac{k_1(x)\chi}{2}\right] u^2\right)\id y \nonumber\\
&& \label{inter3}
\end{eqnarray}
where $\theta>0$ is an arbitrary parameter.\\\\
In order to apply the Lax-Milgram theorem (see e.g. \cite{Evans}) to the bilinear operator $B=B_{x,\epsilon,\lambda,\mu}:\hat{H}^\mu\times \hat{H}^\mu \to \mathbb{R}$, we need to show
\begin{lemma}\label{coercivity} (coercivity of $B_{x,\epsilon,\lambda}$)\\
For all $\lambda>\beta^{-1}$, there are positive constants $c_i$ so that
\begin{eqnarray} \label{condition 1}
&& \forall u,v \in \hat{H}^0:|B_{x,\epsilon,\lambda}(u,v)|\leq c_1 \left\|u\right\|_{\hat{H}^0}\left\|v\right\|_{\hat{H}^0} \\
&& \forall u \in \hat{H}^0: B_{x,\epsilon,\lambda}(u,u) \geq c_2\left\|u\right\|_{\hat{H}^0}^2 \label{condition 2}\\
&& \forall u,v \in H^0: |B_{\epsilon}(u,v)| \leq c_3 \left\|u\right\|_{H^0}\left\|v\right\|_{H^0} \label{condition 3} \\
&& \forall u \in H^0: B_{\epsilon}(u,u) \geq c_4 \left\|u\right\|_{H^0}^2 \label{condition 4}\\
&& \forall u,v \in H^1: |\tilde{B}(u,v)| \leq c_5 \left\|u\right\|_{H^1}\left\|v\right\|_{H^1} \label{condition 5} \\
&& \forall u \in H^1: \tilde{B}(u,u) \geq c_6 \left\|u\right\|_{H^1}^2 \label{condition 6}
\end{eqnarray}
Also the following bound (which is not required for the Lax-Milgram theorem but rather later to meet the requirements of the Hille-Yosida theorem) applies: $\exists c_7\in \mathbb{R}$: $\forall u \in \hat{H}^0$
\begin{equation} \label{bound to B 2}
(c_7+\lambda) \left\|u\right\|_{\mathscr{L}^2}^2 \leq B_{x,\epsilon,\lambda}(u,u)
\end{equation}
%In fact, $\exists c_8>0$ s.t. for all $u \in H^0$
%\begin{equation} \label{bound to B 3}
%c_8 \left\|u\right\|_{\mathscr{L}^2}^2 \leq B_{x,\epsilon,\lambda}(u,u).
%\end{equation}
\end{lemma}
\begin{proof}
First, calculate
\begin{eqnarray*}
&& |B_{x,\epsilon,\lambda}(u,v)|= \\
&& \left|\int \rho_0\left(\beta^{-1}\nabla u\cdot\nabla v+u\left(\nabla [U(x)- U(0)]+\epsilon F\right)\cdot \nabla v+\left[\lambda+\beta^{-1}\mu\frac{1}{|y|^2}\right] uv \right)\id y \right| \nonumber\\
&& \leq \beta^{-1}\left|\int \rho_0\left(\nabla u\cdot\nabla v+\left[\beta\lambda+\mu\frac{1}{|y|^2}\right] uv\right)\right|\\
&&\qquad + \int \rho_0\left|\nabla [U(x)- U(0)]+\epsilon F\right||u\nabla v|\id y. \label{inter}
\end{eqnarray*}
When $x=0$, $\lambda$ and $v$ is radial ($F^i\partial_iv=0$), the second term is zero and we can bound \eqref{inter} by $\|u\|_{H^{\mu}}\|v\|_{H^{\mu}}$, thus obtaining \eqref{condition 5}. For the cases \eqref{condition 1} and \eqref{condition 3}, we proceed with \eqref{inter} (where we now put $\mu=0$)
\begin{eqnarray*}
&& \eqref{inter} \leq \max\left(\beta^{-1},\lambda\right)\left\|u\right\|_{\hat{H}^0}\left\|v\right\|_{\hat{H}^0} +\int \rho\left|\nabla_y [U(x,y)- U(0,y)]+\epsilon F\right||u\nabla v|\id y\\
&& \leq \max\left(\beta^{-1},\lambda\right)\left\|u\right\|_{\hat{H}^0}\left\|v\right\|_{\hat{H}^0}+\int\rho\left(|\nabla [U(x)- U(0)]|+\epsilon | F|\right)|u\nabla v|\id y \\
&& \leq \max\left(\beta^{-1},\lambda\right)\left\|u\right\|_{\hat{H}^0}\left\|v\right\|_{\hat{H}^0} + \left(k_1(x)+\epsilon C\right)\left\|u\right\|_{\mathscr{L}^2}\left\|v\right\|_{H^0} \\
&& \leq \left(\max\left(\beta^{-1},\lambda\right)+k_1(x)+\epsilon C)\right)\left\|u\right\|_{\hat{H}^0}\left\|v\right\|_{\hat{H}^0}
\end{eqnarray*}
thus obtaining \eqref{condition 1} and also \eqref{condition 3} after noting that the Poincar\'e inequality implies that $\|u\|_{\hat{H}^0}^2 \leq (1+c^{-2})\|u\|_{H^0}^2$. \\
Next, check \eqref{inter2} for $x=\lambda=0$ and $\chi=1$. The expression simply reduces to $\beta^{-1}\|u\|_{H^\mu}$ so that \eqref{condition 4} and \eqref{condition 6} is recovered. Otherwise we set $\mu=0$, $\chi=1$ and evaluate \eqref{inter3} for $\theta \to \infty$. The result:
\begin{equation*}
B_{x,\epsilon,\lambda}(u,u)\geq \int_M \rho_0\left(\left[\beta^{-1}-\frac{k_1(x)}{2}\right](\nabla u)^2+\left[\lambda-\frac{k_1(x)}{2}\right] u^2\,\right)\id y \label{RHS}
\end{equation*}
To obtain \eqref{condition 2}, notice that \eqref{RHS} $\geq (\min(\beta^{-1},\lambda)-\frac{k_1(x)}{2})\left\|u\right\|_{\hat{H}^0}$ which is positive for $x$ sufficiently small. To obtain \eqref{bound to B 2}, notice that \eqref{RHS} $\geq (\lambda-\frac{k_1(x)}{2})\left\|u\right\|_{\mathscr{L}^2}$ provided $x$ is small enough that $k_1(x)<2\beta^{-1}$.
\end{proof}
\begin{lemma} \label{ext}(Lax-Milgram: the boundedness of the `external' linear functional)\\
The linear functional $f_v: u \mapsto (u,v)_{\mathscr{L}^2}$ is bounded in the case that
\begin{enumerate}
\item Dom$(f_v)=\hat{H}^0$ and $v \in \mathscr{L}^2$.
\item Dom$(f_v)=H^0$ and $v \in \mathscr{L}^2$.
\item Dom$(f_v)=H^1$ and $v \in \mathscr{L}^2$ is radial.
\end{enumerate}
In case 1 and 3, the map $v \mapsto f_v$ is one-to-one. In case 2, the kernel of that map are the constant functions $\langle 1\rangle$.
\end{lemma}
\begin{proof}
One has 
\begin{equation} \label{starters}
|(u,v)_{\mathscr{L}^2}| \leq \|u\|_{\mathscr{L}^2}\|v\|_{\mathscr{L}^2} .
\end{equation}
In general, the RHS is smaller than $\|u\|_{\hat{H}^0}\|v\|_{\mathscr{L}^2}$ which checks the first case. For the second case, one uses the Poincar\'e inequality to see that the RHS of \eqref{starters} is smaller than $c^{-1}\|u\|_{H^0}\|v\|_{\mathscr{L}^2}$. In the third case, one uses the isometry \eqref{isometry} to obtain
\[|(u,v)_{\mathscr{L}^2}|= |(u\hat{y}^i,\,v\hat{y}_i)_{\mathscr{L}^2}|\leq \|u\hat{y}_i\|_{\mathscr{L}^2}\|v\hat{y}_i\|_{\mathscr{L}^2}\leq c^{-1}\|u\hat{y}\|_{H^0}\|v\hat{y}\|_{\mathscr{L}^2}= c^{-1}\|u\|_{H^1}\|v\|_{\mathscr{L}^2}.\]
Concerning the injectivity of $v \mapsto f_v$, it follows in case 1 and 3 from the density of $\hat{H}^0$ in $\mathscr{L}^2$ and $H^1$ in the radial $\mathscr{L}^2$-functions respectively.
\end{proof}
Combining Lemmae \ref{coercivity} and \ref{ext}, we can use the Lax-Milgram Theorem to define the resolvents
\begin{definition}
For $\lambda>0$, we define the resolvents $R_{x,\epsilon,\lambda}: \mathscr{L}^2 \to (\hat{H}^0 \subset)\mathscr{L}^2$ so that $w=R_{x,\epsilon,\lambda}v \in \hat{H}^0 \subset \mathscr{L}^2$ is the unique solution of the weak equation (which must hold for all $u \in \hat{H}^0$)
\begin{equation*}
B_{x,\epsilon,\lambda}(u,w)=(u,v)_{(\mathscr{L}^2)}.
\end{equation*}
Similarly, for $\lambda=0$, $x=0$ and $w=0$ or $1$, we define (respectively) the pseudo-inverses $R_{\epsilon},\tilde{R}: \mathscr{L}^2 \to \mathscr{L}^2$ so that $w=R_{\epsilon}v \in H^0 \subset \mathscr{L}^2$ and $z=\tilde{R}v \in H^1 \subset \mathscr{L}^2$ is the unique solution of the weak equation (which must hold for all $u \in H^0$ resp. $\in H^1$)
\begin{equation*}
B_{\epsilon}(u,w)=-(u,v)\qquad \tilde{B}(u,z)=-(u,v)
\end{equation*}
\end{definition}
\begin{prop} \label{resolvent properties}(properties of the resolvents and pseudo-inverses)
\begin{enumerate}
\item The resolvents $R_{x,\epsilon,\lambda}$, $R_{\epsilon}$ and $\tilde{R}$ are one-to-one.
\item Considering the operator-norms of the resolvents, one has ($\forall \lambda > c_7$)
\begin{eqnarray} 
&&\sup_{u \in \mathscr{L}^2} \frac{\left\|R_{x,\epsilon,\lambda}u\right\|_{\mathscr{L}^2}}{\left\|u\right\|_{\mathscr{L}^2}}\leq \frac{1}{c_7+\lambda}\label{second bound 1} \\
&& 
\sup_{u \in \mathscr{L}^2} c\frac{\left\|R_{\epsilon}u\right\|_{\mathscr{L}^2}}{\left\|u\right\|_{\mathscr{L}^2}}\leq \sup_{u \in \mathscr{L}^2} \frac{\left\|R_{\epsilon}u\right\|_{H^0}}{\left\|u\right\|_{\mathscr{L}^2}} \leq c_4^{-1} \label{second bound 2}
\end{eqnarray}
where $k_4(x)>0$ is a constant independent of $\epsilon$.
Similar bounds hold for the operator norms of the pseudo-inverse $R_{\epsilon}$, $\tilde{R}$
\item There exists an unbounded operator $L_{x,\epsilon}:D\subset \mathscr{L}^2 \to \mathscr{L}^2$ whose resolvent is precisely $(\lambda I-L_{x,\epsilon})^{-1}=R_{x,\epsilon,\lambda}$. One has $\ker L_{x,\epsilon}=\langle 1 \rangle$.
\item $R_{x,\epsilon,\lambda}$ and $\tilde{R}$ map non-zero, non-negative functions to positive functions.
\item One has 
\begin{equation} \label{unit im}
R_{x,\epsilon,\lambda}1=\lambda^{-1}1.
\end{equation}
\item One has the identities
\begin{equation*}
\forall u \in \mathscr{L}^2:\,L_{0,\epsilon}R_{\epsilon}u=P_Hu \qquad \forall u \in D:\, R_{\epsilon}L_{0,\epsilon}=P_Hu.
\end{equation*}
\item The following adjoint properties hold when $x=0$:
\begin{equation*}
\forall u,v\in \mathscr{L}^2:\,(u,R_{\epsilon}v)=(R_{-\epsilon}u,v),\qquad \forall u,v\in D:\,(u,L_{0,\epsilon}v)=(L_{0,-\epsilon}u,v)
\end{equation*}
\item When $v \in H^0$ is twice weakly-differentiable and $\ell_{x,\epsilon}v \in \mathscr{L}^2$, then $v \in D$ and
\begin{equation*}
L_{x,\epsilon}v=\ell_{x,\epsilon}v
\end{equation*}
\end{enumerate}
\end{prop}
\begin{proof}
\begin{enumerate}
\item In the case of $R_{x,\epsilon,\lambda}$ and $\tilde{R}$, this follows from the Lax-Milgram Theorem and the fact that the functional $v \mapsto f_v$ ($f_v:u \mapsto (u,v)$) is one-to-one (see Lemma \ref{ext}). The case of $R_{\epsilon}$ is similar (in this case $f_v$ is the zero functional only when $v$ is a multiple of $1$).
\item \eqref{second bound 1} follows from \eqref{condition 2} and \eqref{bound to B 2}. Indeed, for all $u \in \mathscr{L}^2$ and denoting $v=\tilde{R}_{x,\epsilon,\lambda}u$
\begin{equation*}
(v,\,u)_{\mathscr{L}^2}=B_{x,\epsilon}(v,\,v)\geq (c_7+\lambda)\left\|v\right\|_{\mathscr{L}^2}^2
\end{equation*}
\eqref{second bound 2} follows from \eqref{condition 4} and the Poincar\'e inequality in a similar fashion.
\item Let us first define the operator $L_{\epsilon}:(R_{\epsilon}^{-1}(\mathscr{L}^2)+\langle1\rangle)\subset \mathscr{L}^2 \to \mathscr{L}^2$ by requiring
\[\begin{cases}
& \forall u \in H:\,L_{\epsilon}R_{\epsilon}u=u \\
& L_{\epsilon}1=0 \\
& L_{\epsilon} \text{ is linear.}
\end{cases}\]
It is easy to see that one then has for all $u \in D:=(R_{\epsilon}^{-1}(\mathscr{L}^2)+\langle1\rangle)$ that $R_{\epsilon}L_{\epsilon}u=u-1(u,1)=P_Hu$. It is proved in item 8 that $D$ includes all smooth and compactly supported functions which are dense in $\mathscr{L}^2$. Next, we define $L_{x,\epsilon}:D \subset \mathscr{L}^2 \to \mathscr{L}^2$ through the action
\[L_{x,\epsilon}v=L_{\epsilon}u-\nabla_y[U(x)-U(0)]\cdot \nabla v.\]
which maps to $\mathscr{L}^2$ because of \ref{lip} and $D$-functions are $H^0$-integrable. Notice that $\forall u \in \hat{H}^0$
\[(u,(\lambda I-L_{x,\epsilon})v)=B_{\epsilon}(u,v)+\int_M\rho_0 \,\left(u\partial_i(U(0)-U(x))\partial^iv+\lambda uv\right)\,\id y=B_{x,\epsilon,\lambda}(u,v)\]
so that $R_{x,\epsilon,\lambda}(\lambda I-L_{x,\epsilon})v=v$. Since $R_{x,\epsilon,\lambda}$ is one-to-one and bounded, this means that $R_{x,\epsilon,\lambda}$ is the resolvent of $L_{x,\epsilon}$.

\item Let us fix a nested sequence $(C_n\subset M)_n$ of compact subsets of $M$ with a smooth boundary (e.g. origin-centered disks with radius $<R$) such that $\cup_n C_n = M$. On the disk $C_n$ let us solve the auxiliary boundary-value problem
\begin{equation*}
\begin{cases}
& (\lambda + \frac{\mu }{\beta|y|^2})- \ell_{x,\epsilon} u_n = v\\
& \lim_{y\to x \in \partial C_n} u_n(y)=0 \\
& \text{ if $\mu\neq 0$: }\lim_{y\to 0} u_n(y)=0
\end{cases}
\end{equation*}
for a given smooth, compactly supported and non-negative $v\neq 0$ (we denote $r=\sup_{y \in \text{ supp }v} |y|$). This problem is standard and a strong $C^{\infty}$ solution $u_n$ exists provided the coefficients of $\ell_{x,\epsilon}$ are smooth (\cite{Qian} p.70-73, \cite{Evans}). In addition, the maximum principle implies that $u_n \geq 0$ (and is zero only at the boundary $\partial C_n$ provided $C_n \cap \text{supp }v \neq \emptyset$) and 
\begin{equation}\label{comp}
\left\|u_n\right\|_{\infty}\leq \frac{1}{\lambda+\frac{\mu}{\beta r^2}}\|v\|_{\infty}.
\end{equation}
Moreover, the same maximum principle implies that $(u_n)_n$ is non-decreasing in $n$ and therefore converges point-wisely to a positive function $u_{\infty}$. Let us presently continue with the case $\mu=0$ (and prove the statement about $R_{x,\epsilon,\lambda}$, the statement about $\tilde{R}$ is proved similarly).\\ It is a straightforward calculation to verify that $u_{\infty}$ still solves $(\lambda I-\ell_{x,\epsilon})u_{\infty}=v$ in the weak sense. Elliptic regularity then again implies that $u_{\infty}\in C^{\infty}$. Now the equation
\[B_{x,\epsilon,\lambda}(\xi,u_{\infty})=(\xi,v)\]
holds for every smooth, compactly supported $\xi$. But that set of functions is dense in $\hat{H}^0$ (Lemma \ref{density}), so by the definition of $R_{x,\epsilon,\lambda}$ and \eqref{condition 1} we can conclude that $u_\infty = R_{x,\epsilon,\lambda}v$ after verifying that $u_\infty \in \hat{H}^1$. But $\chi_n u_\infty$ (with $\chi_n$ specified as in Lemma \ref{comp supp}) is smooth and compactly supported, while by inspection of the expression \eqref{inter3} one can see $ B_{x,\epsilon,\lambda}(\chi_n u_\infty,u_\infty) \to B_{x,\epsilon,\lambda}(u_\infty,u_\infty)+\frac{\theta}{2}(u_\infty,u_\infty)$. So we let $\theta \to 0$ and \eqref{condition 2} implies
\begin{eqnarray*}
&& c_2\|u_{\infty}\|_{\hat{H}^0}^2 \leq B_{x,\epsilon,\lambda}(u_\infty,u_\infty)=\limsup_{n\to \infty}B(\chi_n u_\infty,u_\infty)=\limsup_{n\to \infty}(\chi_n u_\infty,v) \\
&&\leq \|\chi_n u_\infty\|_{\mathscr{L}^2} \|v\|_{\mathscr{L}^2} \leq \| u_\infty\|_{\infty} \|v\|_{\infty} \leq \lambda^{-1}\|v\|_{\infty}^2
\end{eqnarray*}
(where we used \eqref{comp} for the last inequality). So $R_{x,\epsilon,\lambda}v=u_\infty>0$. This finishes the proof that $R_{x,\epsilon,\lambda}$ maps smooth, compactly supported, non-negative functions to positive functions. Since $R_{x,\epsilon,\lambda}$ is a continuous operator and because of Lemma \eqref{density}, the assumption of $v$ being smooth and compactly supported can be dropped by a density argument: A non-zero, non-negative function $v \in \mathscr{L}^2$ is the $\mathscr{L}^2$-limit of an increasing sequence $(v_n)_n$ of non-negative, smooth, compactly supported functions. Since $R_{x,\epsilon,\lambda}v_n$ is, by the foregoing, itself an increasing sequence of positive functions, the limit $R_{x,\epsilon,\lambda}v$ is positive as well.
\item Proved already during the proof of item 3.
\item Proved already during the proof of item 3.
\item One has (for all $u \in \mathscr{L}^2$, $v \in H^0$)
\[\forall u \in \mathscr{L}^2, \,v \in H^0:\, (R_{\epsilon}u,v)=B_{-\epsilon}(R_{\epsilon}u,R_{-\epsilon}v)\underbrace{=}_{\eqref{symmetry}}B_{\epsilon}(R_{-\epsilon}v,R_{\epsilon}u)=(R_{-\epsilon}v,u)\]
But $H^0$ is dense in $H$ and $(R_{\epsilon}u,1)=0=(1,R_{-\epsilon})$ so that $(R_{\epsilon}u,v)=(u,R_{-\epsilon}v)$ holds true for all $v \in \mathscr{L}^2$
\[\forall u,v \in D:\,(u,(\lambda I - L_{0,\epsilon})v)=B_{0,\epsilon,\lambda}(u,v)=B_{0,-\epsilon,\lambda}(v,u)=(u,(\lambda I - L_{0,-\epsilon})v)\]
\item It follows from \eqref{partial int} that for all $u$ smooth and compactly supported
\begin{equation} \label{int}
(u,\,\ell_{0,\epsilon})v)=-B_{\epsilon}(u,v).
\end{equation}
But since the compactly supported smooth functions are dense in $\hat{H}^0$ and because of \eqref{condition 1}, \eqref{int} holds for all $u \in \hat{H}^0$. In particular, 
\begin{enumerate}
\item \eqref{int} holds for all $u \in H^0$. 
\item $(1,\,\ell_{0,\epsilon}v)=-B_{\epsilon}(1,v)=0$
\end{enumerate}
therefore $\ell_{0,\epsilon}v=L_{0,\epsilon}v$. But through the way we have defined $L_{x,\epsilon}$, also the equality $\ell_{x,\epsilon}v=L_{x,\epsilon}v$ follows.
\end{enumerate}
\end{proof}
\begin{Corollary} \label{semigroup}(semigroup)
	The operator $L_{x,\epsilon}$ generates a positive, strongly-$\mathscr{L}^2$-continuous contraction semi-group $(S_{x,\epsilon}(t))_{t\geq 0}$. In addition, probability-conservation holds:
	\begin{equation*}
	S(t)1\equiv 1
	\end{equation*}
\end{Corollary}
\begin{proof}
	From \eqref{second bound 1} we deduce that the resolvent of $L_{x,\epsilon}$, is bounded in accordance with the assumptions of the Hille-Yosida Theorem for strongly-continuous contraction semigroups. That implies the existence and uniqueness of the generated semigroup $S$. Positivity of this semigroup (i.e. $\forall t \geq 0:\,f(\in \mathscr{L}^2)\geq 0 \Rightarrow S_{x,\epsilon}(t)f\geq0$) follows from the Hille-Yoshida identity
	\begin{equation*}\label{HY id}S_{x,\epsilon}(t)=\lim_{\lambda \to \infty} e^{-\lambda t}e^{\lambda^2tR_{x,\epsilon,\lambda}}\end{equation*}
	the positivity of the resolvents (item 2 of Theorem \eqref{resolvent properties}) and the Taylor-formula for the exponential of a bounded operator. The probability-conservation property follows again from \eqref{unit im}, \eqref{HY id} and the Taylor formula for the exponential.
\end{proof}
\begin{lemma} \label{density}
The set of smooth, compactly supported functions is dense in $\hat{H}^0$.
\end{lemma}
\begin{proof}
It suffices to prove the following three steps:
\begin{enumerate}
\item For every $u \in \hat{H}^0$, there is a sequence $(u_n)_n$ of bounded $\hat{H}^0$-functions converging to $u$.
\item For every bounded $\hat{H}^0$-function $v$ there is a sequence of compactly supported $\hat{H}^0$-functions $(v_n)$ converging to $v$.
\item A compactly supported $\hat{H}^0$-function can be approximated by smooth counterparts.
\end{enumerate}
The first step is proved as follows: consider a sequence $\kappa_n: \mathbb{R} \to \mathbb{R}$ of compactly supported functions with $\kappa_n \to Id: s \mapsto s$ as $n \to \infty$ and moreover $s\kappa_n(s)>0$ unless $s=0$. This sequence is tuned in such a way that $s \mapsto (\kappa_n'(s)-1)^2$ is uniformly bounded by some $K>0$. Then for any $u \in \hat{H}^0$, $u_n:=\kappa_n \circ u$ is bounded and
\begin{equation} \label{the end}
\left\|u_n-u\right\|_{\hat{H}^0}^2 = \int_M \rho_0 \left\{(\kappa_n'(u)-1)^2(\nabla u)^2+(\kappa_n(u)-u)^2\right\}\id y
\end{equation}
The integrand in the RHS is pointwise bounded by the integrable function $K(\nabla u)^2+u^2$, so that the dominated convergence theorem implies that \eqref{the end} converges to zero.\\
The second step is proved as follows: approximate the bounded $\hat{H}^0$-function $u$ by $\chi_nu$ where $(\chi_n)_n$ is the sequence constructed in Lemma \eqref{comp supp}. It is straightforward to prove that $\|u-\chi_nu\|_{\hat{H}^0} \to 0$.\\
The third step is standard and treated in many standard PDE textbooks.\\
Virtually the same proof applies for $H^1$ after one identifies that space with the radial vectorfields with components in $H^0$ (lemma \eqref{isom lemma}). The only required modification is that in \eqref{the end} one uses the $\|\|_{H^0}$ norm so that the second term of the RHS is not there.
\end{proof}

\begin{prop} \label{equiv}
Let $\varphi$ be a radial $\mathscr{L}^2$-function, then $f^i=\varphi \hat{y}^i$ is in $\mathscr{L}^2$ and ($\forall \epsilon \in \mathbb{R}$)
\[\hat{y}_i(R_{0}f^i)(y)\equiv(\tilde{R}\varphi)(y)\]
\end{prop}

\begin{proof}
Pick any $H^1$-function $u$, by lemma \ref{isom lemma} $\tilde{R}\varphi \in H^1 \Rightarrow (\tilde{R}\varphi)\hat{y}^i \in H^0$ and therefore
\begin{eqnarray*}
&& \tilde{B}(u,\,\tilde{R}\varphi)=(u,\varphi)=(u\hat{y}_i,\,\varphi\hat{y}^i)=B_{0}(u\hat{y}_i,\,R_{\epsilon}(\varphi\hat{y}^i ))=\beta^{-1}(u\hat{y}_i,\,R_{\epsilon}(\varphi\hat{y}^i ))_{H^0}\\
&&\underbrace{=}_{\eqref{isometry}}\beta^{-1}(u ,\,\hat{y}_iR_{\epsilon}(\varphi\hat{y}^i ))_{H^1}=\tilde{B}(u,\,\hat{y}_iR_{\epsilon}(\varphi\hat{y}^i )),
\end{eqnarray*}
so that
\begin{equation*}
\tilde{B}(u,\,[(\tilde{R}\varphi)-\hat{y}_i(R_{\epsilon}f^i)])=0
\end{equation*}
Taking $u=\tilde{R}\varphi-\hat{y}_i(R_{\epsilon}f^i)$ (which is radial of course) and considering \eqref{condition 6} then yields the desired result.
\end{proof}

\subsection{Proof of Theorem \ref{sta}} \label{proof}
\begin{proof}
Let us first prove that $\nu_{x,\epsilon}$ is well-defined and in $H^0 \cap \mathscr{L}^2$. Note that
\[\left\|L_{0,-\epsilon}^{-1}F_x\right\|_{H^0}:=\sup_{u\in H^0}\frac{\left\|L_{0,-\epsilon}^{-1}F_x u\right\|_{H^0}}{\left\|u\right\|_{H^0}}\leq \left(\sup_{u\in \mathscr{L}^2}\frac{\left\|L_{0,\epsilon}^{-1} u\right\|_{H^0}}{\left\|u\right\|_{\mathscr{L}^2}}\right)\left(\sup_{v\in H^0}\frac{\left\|F_x v\right\|_{\mathscr{L}^2}}{\left\|v\right\|_{H^0}}\right) \leq k_4(x)k_F(x)=:k_S(x).\]
This together with $L_{-\epsilon}^{-1}F_x 1 \in H^0$ implies that
\[\left\|(L_{-\epsilon}^{-1}F_x)^n 1\right\|_{H^0} \leq (k_S(x))^{n-1}\left\|L_{-\epsilon}^{-1}F_x 1\right\|_{H^0}.\]
So define the radius $R_{\rho}>0$ so that $\forall x \in B(0,\epsilon): |x|<R_{\rho}\Rightarrow k_S(x)<1$. Then for $x \in B(0,R_{\rho})$, the partial sums associated to the sequence \eqref{stat dens} constitute a Cauchy sequence in $H^0$. $H^0$ is a Hilbert space and therefore the limit $\nu_{x,\epsilon} \in H^0$. By the Poincar\'e inequality \eqref{Poincare inequality} $\nu_{x,\epsilon} \in \mathscr{L}^2$. \\
To prove the stationarity \eqref{stat}, it suffices to prove the more general identity
\begin{equation}\label{to prove}B_{x,\epsilon}(\nu_{x,\epsilon},u)=0\end{equation}
for all $u \in \hat{H}^0$. Note that $F_x$ was defined in such a way that when $u$ is smooth and compactly supported and $v \in H^0$, then
\[B_{x,\epsilon}(v,u)-B_{0,\epsilon}(v,u)=-(F_xv,u)\]
Therefore, sticking to $u$ smooth and compactly supported
\begin{eqnarray*}
&& B_{x,\epsilon}(\nu_{x,\epsilon},u)=B_{x,\epsilon}\left(\lim_{N \to \infty} \sum_{n=0}^N(L_{-\epsilon}^{-1}F_x)^n 1,\,u\right)=\lim_{N \to \infty} \sum_{n=0}^NB_{x,\epsilon}\left((L_{-\epsilon}^{-1}F_x)^n 1,\,u\right) \\
&& = \lim_{N \to \infty}\sum_{n=0}^N \left\{B_{0,\epsilon}\left((L_{-\epsilon}^{-1}F_x)^{n} 1,\,u\right)-\left(F_x(L_{-\epsilon}^{-1}F_x)^n 1,\,u\right)\right\}\\
&& = \lim_{N \to \infty} \left\{\sum_{n=1}^NB_{0,-\epsilon}\left(u,\,(L_{-\epsilon}^{-1}F_x)^{n} 1\right)-\sum_{n=0}^N\left(F_x(L_{-\epsilon}^{-1}F_x)^n 1,\,u\right)\right\} \\
&& =\lim_{N \to \infty} \left\{\sum_{n=1}^N\left(u,\,(F_xL_{-\epsilon}^{-1})^{n-1}F_x 1\right)-\sum_{n=0}^N\left(F_x(L_{-\epsilon}^{-1}F_x)^n 1,\,u\right)\right\} \\
&& =-\lim_{N \to \infty} (F_x(L_{-\epsilon}^{-1}F_x)^{N} 1,\,u) = 0.
\end{eqnarray*}
Since \eqref{to prove} holds for all $u$ smooth and compactly supported, lemma \ref{density} and \eqref{condition 1} imply that the identity holds for all $u \in \hat{H}^0$.\\
Finally, to check \eqref{pert}, note that
\[\left\|\nu_{x,\epsilon}-1-L_{-\epsilon}^{-1}F_x 1\right\|_{\mathscr{L}^2}\leq c^{-1}\left\|\nu_{x,\epsilon}-1-L_{-\epsilon}^{-1}F_x 1\right\|_{H^0}\leq c^{-1}\sum_{n=2}^{\infty}\left\|(L_{-\epsilon}^{-1}F_x 1)^n\right\|_{H^0} \leq c^{-1}k_S(x)\]
So far we have achieved establishing that $\nu_{x,\epsilon}$ solves \eqref{to prove} and is therefore stationary. To prove that it is unique in this respect,\\
suppose there were another density $\tilde{\rho}_{x,\epsilon}$ such that $\frac{\tilde{\rho}_{x,\epsilon}}{\rho_0}\in \mathscr{L}^2$ and
\begin{equation}\left(\frac{\tilde{\rho}_{x,\epsilon}}{\rho_0},\,L_{x,\epsilon}u\right)\end{equation}
for all $u \in D$.\\
Then $\exists \lambda \in \mathbb{R}$ such that 
\[\nu_{x,\epsilon}:=\frac{\tilde{\rho}_{x,\epsilon}}{\rho_0} + \lambda \frac{\rho_{x,\epsilon}}{\rho_0} \in H\]
Suppose (seeking a contradiction) that $\nu_{x,\epsilon}\neq 0$ and define $v=L_{\epsilon}^{-1}\nu_{x,\epsilon}$. Then $\nu_{x,\epsilon}= L_\epsilon v$ and therefore
\[(\nu_{x,\epsilon},L_{\epsilon}v)=(\nu_{x,\epsilon},\nu_{x,\epsilon})>0.\]
But since 
\[\|(L_{x,\epsilon}-L_{\epsilon})v\|_{\mathscr{L}^2}=\left\|\nabla\left( U(x)- U(0)\right)\cdot\nabla v\right\|_{\mathscr{L}^2} \leq k_1(x)\|v\|_{H^0},\]
we then obtain
\begin{eqnarray*}&& (\nu_{x,\epsilon},L_{x,\epsilon}v) \geq (\nu_{x,\epsilon},L_{\epsilon}v)-k_1(x)\|\nu_{x,\epsilon}\|_{\mathscr{L}^2}\|v\|_{H^0}= (\nu_{x,\epsilon},\nu_{x,\epsilon})-k_1(x)\|\nu_{x,\epsilon}\|_{\mathscr{L}^2}\|L_{\epsilon}^{-1}\nu_{x,\epsilon}\|_{H^0} \\
	&&\geq (1-k_1(x)k_4(x))\|\nu_{x,\epsilon}\|_{\mathscr{L}^2}^2.
\end{eqnarray*}
For $x$ sufficiently small, the RHS (and therefore the LHS is positive). But we then arrive at the following contradiction
\[0<(\nu_{x,\epsilon},L_{x,\epsilon}v)=\left(\frac{\tilde{\rho}_{x,\epsilon}}{\rho_0},\,L_{x,\epsilon}v\right)-\lambda \left(\frac{\rho_{x,\epsilon}}{\rho_0},\,L_{x,\epsilon}v\right)=0.\]
\end{proof}


\begin{thebibliography}{9}
	
\bibitem{mal}
S. A. Mallory, C. Valeriani, and A. Cacciuto, Curvature-induced activation of a passive tracer in an active bath. Phys. Rev. E {\bf 90}, 032309 (2014).

\bibitem{yan}
W. Yan and J. F. Brady, The force on a boundary in active matter. J. Fluid Mech. {\bf 785}, R1 (2015).

\bibitem{pre}  
N. Nikola, A. P. Solon, Y. Kafri, M. Kardar, J. Tailleur,  and R. Voituriez, Active Particles with Soft and Curved Walls: Equation of State, Ratchets, and Instabilities. Phys. Rev. Lett. {\bf 117}, 098001 (2016).

\bibitem{kaf}
Y.~Baek, A.P.~Solon, X.~Xu, N.~Nikola, and Y.~Kafri, Generic long-range interactions between passive bodies in an active fluid. arXiv:1709.02281 [cond-mat.stat-mech].
	
\bibitem{tim}	
	C.~Maes and T.~Thiery, The induced motion of a probe coupled to a bath with random resettings.  J. Phys. A: Math. Theor. {\bf 50}, 415001 (2017).
	
\bibitem{epl}
	U.~Basu, P.~de Buyl, C.~Maes and K.~Neto\u{c}n\'{y}, Driving-induced stability with long-range effects. EPL {\bf 115}, 30007 (2016).
	
	
\bibitem{prl}
	U.~Basu, C.~Maes, and K.~Neto\u{c}n\'{y}, How Statistical Forces Depend on the Thermodynamics and Kinetics of Driven Media. Phys. Rev. Lett.  {\bf 114}, 250601 (2015).
	
\bibitem{Sheshka} 
R.~Sheshka, P.~Recho, and L.~Truskinovsky,  Rigidity generation by nonthermal fluctuations. Phys. Rev. E {\bf 93}, 052604 (2016). 
	

\bibitem{njp}
	U.~Basu, C.~Maes and K.~Neto\u{c}n\'{y}, Statistical forces from close-to-equilibrium media. New J. of Phys. {\bf 17}, 115006 (2015).

	
\bibitem{Bakry} 
B.~Bakry, and G.~Cattiaux, A simple proof of the Poincar\'e inequality for a large class of probability measures including the log-concave case. Elect. comm. in Probab. {\bf 13}, 60--66 (2008).


\bibitem{Qian}  
   D.-Q.~Jiang, M.~Qian, M.-P.~Qian,
\textit{Mathematical theory of nonequilibrium steady states}, Lecture Notes in Mathematics 1833, Springer (2002).

\bibitem{Evans} 
 L.C.~Evans, \textit{Partial differential equations}.
  Graduate Studies in Mathematics --- American Mathematical Society; 2nd edition (March 3, 2010).

\bibitem{Liebermann} 
G.M.~Liebermann, \textit{Second order parabolic differential equations}. World Scientific (November 6, 1996).



 
\end{thebibliography}
\end{document}